\documentclass[aps,showpacs,superscriptaddress,floatfix, notitlepage,prl,longbibliography,reprint,unicode=true,colorlinks=true,citecolor=Blue,linkcolor=RubineRed,urlcolor=Blue]{revtex4-2}
\usepackage[T1]{fontenc}
\usepackage[latin9]{inputenc}
\usepackage{color}
\usepackage{bm}
\usepackage{amsmath}
\usepackage{amsthm}
\usepackage{amssymb}
\usepackage{graphicx}
\usepackage{ulem}
\usepackage[]
 {hyperref}

\makeatletter
\theoremstyle{plain}

\theoremstyle{plain}
\newtheorem{lem}{\protect\lemmaname}

\usepackage{braket}
\usepackage{bm}
\usepackage{txfonts} 
\usepackage{graphicx}
\usepackage[usenames,dvipsnames]{xcolor}
\date{\today}
\usepackage{amsthm}
\newtheorem{thm}{Theorem}

\newtheorem{obs}{Observation}

\makeatother

\providecommand{\lemmaname}{Lemma}
\providecommand{\theoremname}{Theorem}

\begin{document}
\title{A geometric criterion for optimal measurements in multiparameter quantum metrology}

\author{Jing Yang}
\email{jing.yang.quantum@zju.edu.cn}
\affiliation{Institute of Fundamental and Transdisciplinary Research, Institute of Quantum Sensing, and Institute for Advanced Study in Physics, Zhejiang University, Hangzhou 310027, China}
\affiliation{Nordita, KTH Royal Institute of Technology and Stockholm University, Hannes Alfv\'ens v\"ag 12, 106 91 Stockholm, Sweden.}

\author{Satoya Imai}
\email{satoyaimai@yahoo.co.jp}
\affiliation{Institute of Systems and Information Engineering, University of Tsukuba, Tsukuba, Ibaraki 305-8573, Japan}
\affiliation{Center for Artificial Intelligence Research (C-AIR), University of Tsukuba, Tsukuba, Ibaraki 305-8577, Japan}

\author{Luca Pezz\`e}
\email{luca.pezze@ino.cnr.it}
\affiliation{Istituto Nazionale di Ottica del Consiglio Nazionale delle Ricerche (CNR-INO), Largo Enrico Fermi 6, 50125 Firenze, Italy}
\affiliation{European Laboratory for Nonlinear Spectroscopy (LENS), University of Florence, 50019 Sesto Fiorentino, Italy}

\begin{abstract}

Determining when the multiparameter quantum Cram\'er--Rao bound (QCRB) is saturable with experimentally relevant single-copy measurements is a central open problem in quantum metrology. Here we establish an equivalence between QCRB saturation and the simultaneous hollowization of a set of traceless operators associated with the estimation model, i.e., the existence of complete (generally nonorthogonal) bases in which all corresponding diagonal matrix elements vanish. This formulation yields a geometric characterization: optimal rank-one measurement vectors are confined to a subspace orthogonal to a state-determined Hermitian span. This provides a direct criterion to construct optimal Positive Operator-Valued Measures(POVMs). We then identify conditions under which the partial commutativity condition proposed in \href{https://link.aps.org/doi/10.1103/PhysRevA.100.032104}{[Phys. Rev. A 100, 032104(2019)}] becomes necessary and sufficient for the saturation of the QCRB, demonstrate that this condition is not always sufficient, and prove the counter-intuitive uselessness of informationally-complete POVMs. 
\end{abstract}
\maketitle

The quantum Cram\'er-Rao bound (QCRB)~\citep{helstrom1976quantum,helstrom1968theminimum} sets the ultimate precision limit for estimating parameters encoded in quantum states and underpins quantum-enhanced sensing and metrology.
This bound proves crucial for quantifying the impact of quantum noise in practical applications such as interferometry, magnetometry, optical imaging, gravitational wave detection, navigation, atomic clocks, and optical phase estimation, see e.g. \cite{degen2017quantum, pezz`e2018quantum}.
While the QCRB is always achievable in single-parameter estimation~\citep{braunstein_statistical_1994,braunstein_generalized_1996}, in multiparameter problems it remains valid but need not be saturable~\citep{albarelli2020aperspective,liu2019quantum,demkowicz-dobrzanski2020multiparameter,suzuki2020quantum,pezz`e2025advances}.
This limitation is intrinsically quantum: the measurements optimal for different parameters may be incompatible, as the corresponding observables need not commute.
The incompatibility problem in multiparameter sensing has attracted growing interest recently~\cite{carollo2020geometry,tsang2020quantum,belliardo2021incompatibility}. 

Compared to other metrological bounds~\citep{holevo2011probabilistic,hayashi2005asymptotic,conlon2021efficient,tsang2020quantum,albarelli2019evaluating}, the QCRB has a closed form in terms of the quantum state, making it particularly convenient for computations.
However, despite recent advances in multiparameter estimation theory, little is known about the optimal measurements that saturate the bound for general states, especially in the experimentally relevant single-copy regime~\citep{matsumoto2002anew,ragy2016compatibility,pezz`e2017optimal,zhu2018universally, yang2019attaining, lu2021incorporating, chen2022information, zhou2025randomized, PezzPRL2025}.
This problem addresses one of the major open questions in quantum information theory~\citep{horodecki2022fiveopen}.
It is also tied to the search for necessary and sufficient saturation conditions, which are still missing in the literature.
In particular, it was shown that the QCRB  can be saturated at the single-copy level only when a quantum state satisfies the partial commutativity condition (PCC)~\citep{yang2019optimal}.
Understanding whether also the converse is true, remains outstanding~\cite{Imai2026}.

In this work, we establish the connection between saturating the QCRB and the  {\it simultaneous hollowization} condition of multiple traceless matrices. 
Here, hollowization refers to the procedure to find a set of complete bases, not necessarily orthogonal, such that the diagonal matrix elements of a traceless operator vanish.
This yields a geometric picture in which optimal measurements are confined to a subspace orthogonal to the space spanned by some linear Hermitian operators solely determined by the density operator, as schematically shown in Fig.~\ref{fig:cen-geo}. 
Our results imply, among other consequences, that information-complete Positive Operator-Valued Measures (POVMs) are generically ineffective for saturating nontrivial multiparameter QCRB. We also provide an efficient constructive procedure for optimal measurements at the level of each measurement outcome.
We finally identify regimes in which PCC becomes sufficient to saturate the QCRB -particularly if the dimension of the Hilbert space is much larger than the rank of the density matrix and the number of estimation parameters, which is particularly relevant for continuous-variable systems.

\textit{{Optimal measurement conditions.}}{---} We consider a state $\rho_{\bm{\lambda}}=\sum_{a=1}^{r}p_{a\bm{\lambda}}\ket{\psi_{a\bm{\lambda}}}\bra{\psi_{a\bm{\lambda}}}$, in spectral decomposition, that depends on $s$ parameters and denote the estimation
parameter by $\bm{\lambda}=(\lambda_{1},\,\cdots,\,\lambda_{s})$, where the coefficients
$p_{a\bm{\lambda}}$ are strictly positive and $r$ is the rank of the density operator. 
The classical Fisher information matrix (CFIM) that characterize the precision limit associated with the POVM measurement $\{E_{\omega}\}_{\omega=1}^{\Omega}$ is indicated as $[F^C]_{ij}\equiv \sum_\omega[F^C_{\omega}]_{ij}$, where $[F^C_{\omega}]_{ij}= \partial_i p(\omega|\bm{\lambda})\partial_j p(\omega|\bm{\lambda})/p(\omega|\bm{\lambda})$ and the probability $p(\omega|\bm{\lambda})=\mathrm{Tr}(\rho_{\bm{\lambda}}E_{\omega})$ is given by the Born rule.
The quantum Fisher information matrix (QFIM) is given by $[F^Q]_{ij}\equiv \sum_\omega[F^Q_{\omega}]_{ij}$, where $[F^Q_{\omega}]_{ij}=\mathrm{ReTr}(\rho_{\bm{\lambda}}L_{i}E_{\omega}L_{j})$ and the symmetric logarithmic derivative is defined as $(L_i\rho_{\bm{\lambda}}+\rho_{\bm{\lambda}}L_i)/2=\partial_i \rho_{\bm{\lambda}}$, where we suppress the dependence of $L_i$  on $\bm{\lambda}$ for simplicity.
Focusing on optimal POVMs we can assume, without loss of generality, that $E_{\omega}$ is of rank-one.
The matrix inequality
\begin{equation}
F^C_{\omega}\preceq F^Q_{\omega}\label{eq:Refined-Helstrom}
\end{equation}
holds for every $\omega$~\citep{yang2019optimal},
where $A \preceq 0 $ means that $A$ is negative-semi-definite.
The QCRB is saturated if and only if Eq.~(\ref{eq:Refined-Helstrom}) holds for each POVM element $E_\omega$.
Without loss of generality, one can consider rank-one POVM with $E_\omega=\ket{\pi_\omega}\bra{\pi_\omega}$, where $\ket{\pi_\omega}$  is not necessarily normalized and therefore $E_\omega$ is not necessarily projective.
The necessary and
sufficient conditions for saturating Eq.\eqref{eq:Refined-Helstrom} are as follows~\citep{yang2019optimal}:
For regular operators where $\text{Tr}(\rho_{\bm{\lambda}}E_{\omega})\neq0$,
the inequality~(\ref{eq:Refined-Helstrom}) is saturated if and only
if
\begin{equation}
\braket{\pi_{\omega}\big|\psi_{a\bm{\lambda}}}=\xi_{\omega,\,i}\braket{\pi_{\omega}\big|L_{i}\big|\psi_{a\bm{\lambda}}},\,\forall a,\,\omega,\,i,\label{eq:Opt-Reg}
\end{equation}
where $\xi_{\omega,\,i}$ is real coefficient independent of $a$.
For null operators where $\text{Tr}(\rho_{\bm{\lambda}}E_{\omega})=0$,
the inequality~(\ref{eq:Refined-Helstrom}) is saturated if and only
if
\begin{equation}
\braket{\psi_{a\bm{\lambda}}\big|L_{i}\big|\pi_{\omega}}=\eta_{\omega,\,ij}\braket{\psi_{a\bm{\lambda}}\big|L_{j}\big|\pi_{\omega}},\,\forall a,\omega,\label{eq:Opt-Null}
\end{equation}
where $\eta_{\omega,\,ij}$ is real coefficient independent of $a$. 

Eqs.~(\ref{eq:Opt-Reg},~\ref{eq:Opt-Null})
allow one to check whether a given POVM can saturate the QCRB bound or not without calculating the CFIM and QFIM explicitly. However, the presence of the unknown
real coefficients $\xi_{\omega,\,i}$, $\eta_{\omega,\,ij}$ still prevents  a challenge for both analytical and numerical search of
the optimal measurements. 

Below, we shall present a new optimal measurement
condition that removes these real coefficients.
\begin{thm}
\textup{\label{thm:Simultaneous=000020Hollowization=000020Theorem}
[Simultaneous Hollowization Theorem] For a given rank-one POVM operator
$E_{\omega}=\ket{\pi_{\omega}}\bra{\pi_{\omega}}$, the inequality~(\ref{eq:Refined-Helstrom})
is saturated if and only if 
\begin{equation}
\braket{\pi_{\omega}\big|W_{ij,\,ab}\big|\pi_{\omega}}=0,\,\forall i,\,j,\,i\neq j,\,a,\,b,\label{eq:W-opt}
\end{equation}
and 
\begin{equation}
\braket{\pi_{\omega}\big|M_{i,\,ab}\big|\pi_{\omega}}=0,\,\forall a,\,b, \label{eq:M-opt}
\end{equation}
where the indices $i,j$ run from $1$ to the number of parameters $s$, while the indices $a,b$ run from $1$ to the rank $r$ of $\rho_{\bm{\lambda}}$, $P_{ab}\equiv\ket{\psi_{a\bm{\lambda}}}\bra{\psi_{b\bm{\lambda}}}$, $M_{i,\,ab}\equiv[L_{i},\,P_{ab}]$ and
$W_{ij,\,ab}\equiv L_{i}P_{ab}L_{j}-L_{j}P_{ab}L_{i}$. 
In other words, inequality~(\ref{eq:Refined-Helstrom}) can be saturated if and only if there exists a POVM measurement basis $\{\ket{\pi_{\omega}}\}$ such that the operators $W_{ij,\,ab}$ and $M_{{i},ab}$ can be simultaneously ``hollowized'', i.e., they can be brought in a matrix form with vanishing diagonal entries (so-called hollow matrices).
As a result, we indicate Theorem~\ref{thm:Simultaneous=000020Hollowization=000020Theorem}
as the hollowization theorem~\footnote{Clearly, Eqs.~(\ref{eq:W-opt}-\ref{eq:M-opt}) are
$U(1)$ gauge invariant under the transformation of
the eigenstates $\ket{\psi_{a\bm{\lambda}}}\to\ket{\psi_{a\bm{\lambda}}}\mathrm{e}^{\mathrm{i}\theta_{a\bm{\lambda}}}$
and the scaling transformation of the measurement basis $\ket{\pi_{\omega}}\to\alpha_{\omega}\ket{\pi_{\omega}}$,
where $\alpha_{\omega}$ is any complex number.}. 
}
\end{thm}

Furthermore, given $E_\omega=\ket{\pi_{\omega}}\bra{\pi_{\omega}}$, if $F^Q_{\omega}\neq0$~\footnote{The case where $F^Q_{\omega}=0$ is trivial, which leads to $F^C_\omega=0$}, one can always choose a given estimation parameter $\lambda_{\bar{i}}$ and a corresponding eigenstate $\ket{\psi_{\bar{a}\bm{\lambda}}}$ satisfying 
$\braket{\psi_{\bar{a}\bm{\lambda}}\big|L_{\bar{i}}\big|\pi_{\omega}}\neq 0,$
such that the number of constraints on the $W$ and $M$ matrices can be significantly reduced as follows: 
\begin{equation}
\braket{\pi_{\omega}\big|W_{\bar{i}j,\,\bar{a}b}\big|\pi_{\omega}}=0,\,\forall j\in[1,s],\,j\neq\bar{i},\,b\in[1,r],\label{eq:W-opt-pair}
\end{equation}
and 
\begin{equation}
\braket{\pi_{\omega}\big|M_{\bar{i},\,\bar{a}b}\big|\pi_{\omega}}=0,\,\forall b\in[1,r].\label{eq:M-opt-pair}
\end{equation}

We call $(\lambda_j,\ket{\psi_{b\bm{\lambda}}})$  \textit{pair-compatible}  with   $(\lambda_i,\ket{\psi_{a\bm{\lambda}}})$ under the rank-one POVM $E_\omega$ if Eqs.~(\ref{eq:W-opt},\ref{eq:M-opt}) hold for the index pairs $(i,a)$ and $(j,b)$ and $E_\omega$, as indicated by a uni-directional arrow in Fig.~\ref{fig:cen-geo}(a).
All the parameters are called compatible under a complete rank-one POVM, if and only if Eq.~(\ref{eq:Refined-Helstrom}) holds, which is equivalent to the fact that for all the pairs of 
parameters and eigenstates are pair-compatible
With this definition, we observe that pairs with the same estimations parameter $(\lambda_i,\ket{\psi_{a\bm{\lambda}}})$ and $(\lambda_i,\ket{\psi_{b\bm{\lambda}}})$ are always compatible under the projective POVM formed by the eigenvectors of the SLD $L_{{i}}$.

As shown in Fig.~\ref{fig:cen-geo}, Eqs.~(\ref{eq:W-opt-pair},\ref{eq:M-opt-pair}) further break the pair compatibility among all the parameters and eigenstates
down into the pair compatibility between any pair with a central pair consisting of an estimation parameter $\lambda_{\bar{i}}$  and a central eigenstate $\ket{\psi_{\bar{a}\bm{\lambda}}}$. For pure states,
there is only one eigenstate, which provides the intuition for the easier construction of optimal measurements when compared to generic
mixed states~\citep{pezz`e2017optimal,matsumoto2002anew}.

Theorem~\ref{thm:Simultaneous=000020Hollowization=000020Theorem}
is our first main result, which leads to many profound consequences in what follows. The proof can be found in the Sec.~\ref{sec:Proof-Hollowization-Theorem} and \ref{sec:Proof-Pair-Hollowization-Theorem} in the Supplemental Material~\citep{SM}. 
Next, upon summing over $\omega$ in these equations, one can see that the set of operators $\{W_{ij,\,ab},\,M_{i,ab}\}$
are traceless. 
While $M_{i,\,ab}$ is traceless by definition,
$W_{ij,ab}$ is traceless if and only if the PCC holds
\begin{equation}
\braket{\psi_{a\bm{\lambda}}\big|[L_{i},\,L_{j}]\big|\psi_{b\bm{\lambda}}}=0,\,\forall i,j,a,b,\,i\neq j.\label{eq:PCC}
\end{equation}
Thus, we thus see that the PCC is necessary for the saturation of the QCRB and it naturally emerges from the simultaneous hollowization Theorem~\ref{thm:Simultaneous=000020Hollowization=000020Theorem}. We note that the pure state version
of Eq~(\ref{eq:M-opt-pair}) also appears in Ref.~\citep{zhou2020saturating,liu2025optimal}
in searching for optimal local measurements in single-parameter quantum
metrology.

\begin{figure}
\begin{centering}
\includegraphics[scale=0.37]{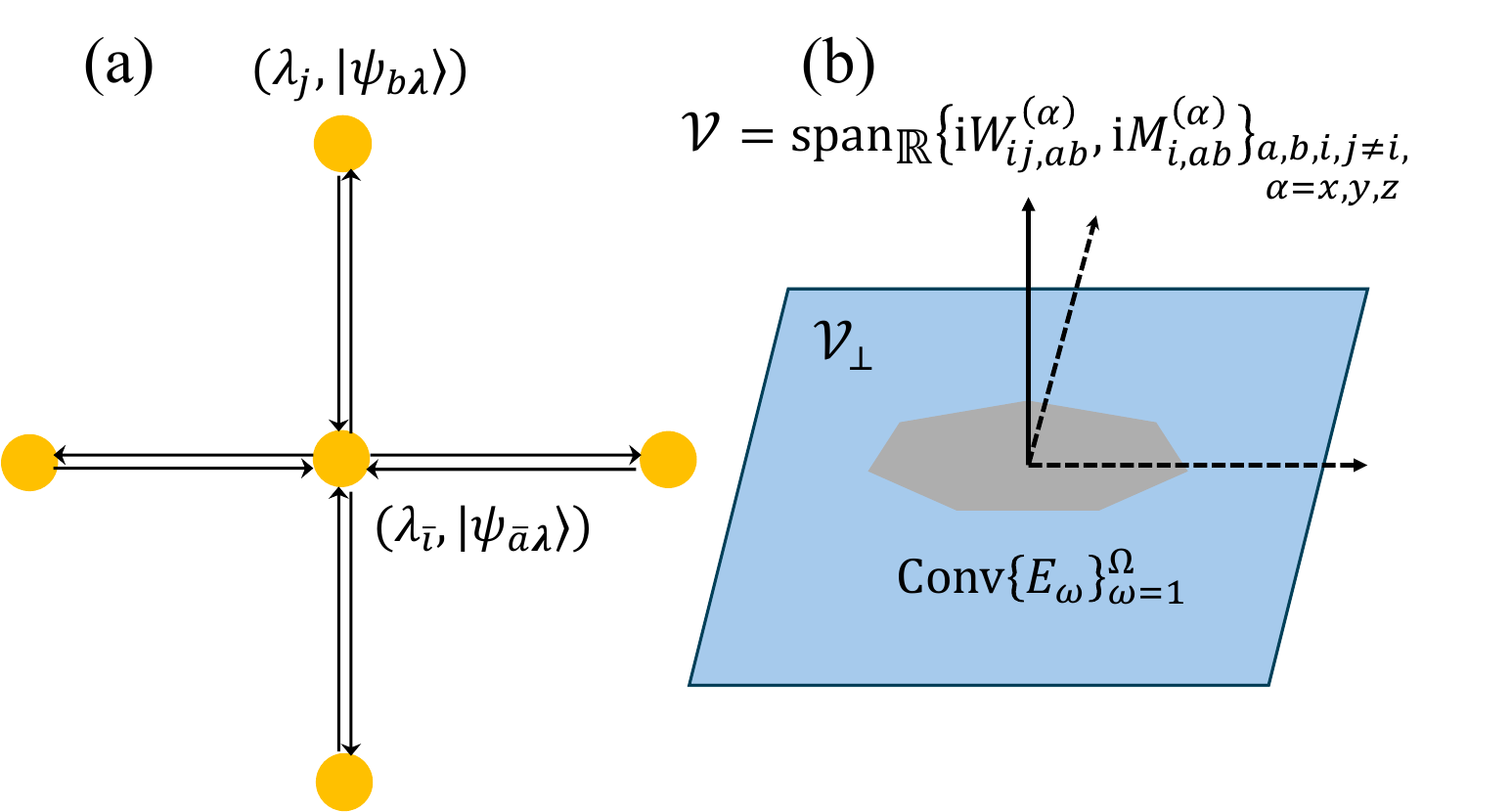}
\par\end{centering}
\caption{\protect\label{fig:cen-geo} (a) Schematic of parameter compatibility under a given rank-one POVM. The yellow dots represent the tuple $(\lambda_j, \ket{\psi_{b\bm{\lambda}}})$. The arrows in both directions implies mutually compatibility. (b) Geometric interpretation of the saturation condition. The gray shaded area is the convex hull formed by the rank-one POVM.
}

\end{figure}

We conclude this section by noting that every single traceless matrix
can be always hollowized~\citep{zhou2020saturating,liu2025optimal,fillmore1969similarity}
through unitary transformation.
As indicated by the 
hollowization Theorems~\ref{thm:Simultaneous=000020Hollowization=000020Theorem}, pair compatibility in multiparameter
estimation would imply the existence of common basis (not necessarily
orthonormal) such that the set of matrices $\{W_{{i}j,\,{a}b},\,M_{{i},\,{a}b}\}$
can be hollowized simultaneously, which highlights the measurement
compatibility problem in multiparameter estimation from a different
perspective. As one might expect, like the standard simultaneous diagonalization,
with the increasing number of the traceless matrices, the difficulty
to hollowize them simultaneously can increase dramatically.

\textit{{Geometry of the saturation condition}}
--- 
In this section, we give a geometric interpretation
of the Hollowization Theorem\textcolor{blue}{~\ref{thm:Simultaneous=000020Hollowization=000020Theorem}}, from which the recipe
for optimal measurements can be easily constructed, if it exists. To this end, let us first focus on the case of finite dimensional
Hilbert space of dimension $d$. we choose a set of
linearly independent rank-one POVM operators $\{E_{\omega}\equiv\ket{\pi_{\omega}}\bra{\pi_{\omega}}\}_{\omega=1}^{|\Omega|}$
such that $\sum_{\omega=1}^{|\Omega|}E_{\omega}=\mathbb{I}$, where
$\Omega$ and $|\Omega|$ are the set and the number of measurement
outcomes, respectively. 
Note that the measurement basis $\{\ket{\pi_{\omega}}\}$
is not necessarily orthonormal. Furthermore, due to the rank requirement
and linear independence, we know $d\leq|\Omega|\leq d^{2}$. 

Next, note that while $W_{ij,aa}$ and $M_{i,\,aa}$
are anti-Hermitian, for $a\neq b$ both $W_{ij,\,ab}$ and $M_{i,\,ab}$
are neither Hermitian nor anti-Hermitian.
To introduce the geometric interpretation of the saturation condition, we construct hermitain matrices from $W$ and $M$  into two anti-Hermtian matrices. To this end, we define
\begin{equation}
W_{ij,\,ab}\equiv\frac{1}{2}[W_{ij,\,ab}^{(x)}+\text{i}W_{ij,\,ab}^{(y)}],\,M_{i,ab}\equiv\frac{1}{2}[M_{i,\,ab}^{(x)}+\text{i}M_{i,\,ab}^{(y)}],\label{eq:W-M-def}
\end{equation}
where $W_{ij,\,ab}^{(\alpha)}\equiv L_{i}\sigma_{ab}^{(\alpha)}L_{j}-L_{j}\sigma_{ab}^{(\alpha)}L_{i}$
and $M_{i,\,ab}^{(\alpha)}\equiv[\sigma_{ab}^{(\alpha)},\,L_{i}]$
with $\alpha=x,\,y,z$ are anti-Hermitian operators and $\sigma_{ab}^{(\alpha)}$
is defined as $\sigma_{ab}^{(x)}\equiv P_{ab}+P_{ba}$, $\sigma_{ab}^{(y)}\equiv-\text{i}(P_{ab}-P_{ba})$.
Furthermore, we define $W_{ij,aa}^{(z)}\equiv L_{i}\sigma_{aa}^{(z)}L_{j}-\mathrm{h.c.}$
and $M_{i,\,aa}^{(z)}\equiv[\sigma_{aa}^{(z)},\,L_{i}]$, where $\sigma_{aa}^{(z)}\equiv P_{aa}$.

With above notations, the hollowization conditions Eq.~(\ref{eq:W-opt})
and Eq.~(\ref{eq:M-opt}) can be straightforwardly recast into
\begin{equation}
\langle\mathrm{i}W_{{i}j,\,{a}b}^{(\alpha)},\,E_{\omega}\rangle=\langle\mathrm{i}M_{{i},\,{a}b}^{(\alpha)},\,E_{\omega}\rangle=0\label{eq:Tr-condition}
\end{equation}
for all $j\in[1,s],\,j\neq i$, $b\in[1,r]$, $\alpha=x,y,z$
and $\omega$, where the inner product defined on the real vector
space $\mathrm{Herm}(d)$ is defined as $\langle A,\,B\rangle=\mathrm{Tr}(AB)$. As a higher-rank
optimal POVM measurement operator can be further refined into several
optimal rank-one POVM measurements, Eq.~(\ref{eq:Tr-condition})
also holds for higher-rank optimal POVM measurement, but is not sufficient
in general.
We denote 
\begin{subequations}
\begin{equation}
\mathcal{\mathcal{V}}\equiv\mathrm{span}_{\mathbb{R}}\{\mathrm{i}W_{{i}j,\,{a}b}^{(\alpha)},\,\mathrm{i}M_{{i},\,{a}b}^{(\alpha)}\}_{a,\,b,\,i,\, j\neq {i},\alpha=x,y,z},
\end{equation}
\end{subequations} 
where
$\mathcal{V}$ is a subspace $\mathrm{Herm}(d)$
or equivalently the vector space $\mathrm{i}\mathcal{V}$ with anti-Hermitian
elements is a subspace of $\mathfrak{u}(d)$. 

As shown in Fig.~\ref{fig:cen-geo}(b), Eq.~(\ref{eq:Tr-condition})
leads to the remarkable geometric interpretation of the saturation
condition: The optimal rank-one POVM operator $E_{\omega}$ must lie
in orthogonal complement of $\mathcal{V}$ in $\mathrm{Herm}(d)$,
denoted as $\mathcal{V}_{\perp}$. 
We introduce the convex hull of a  rank-one POVM measurement, denoted as $\mathrm{Conv}\{E_{\omega}\}_{\omega=1}^{|\Omega|}$, which is defined as the convex combination of the all rank-one POVM measurement. The saturation of the QCRB also means that
\begin{equation}
\mathrm{Conv}\{E_{\omega}\}_{\omega=1}^{|\Omega|}\subseteq\mathcal{V}_{\perp}.
\end{equation}
Since $\sum_{\omega}E_{\omega}=\mathbb{I}$ therefore, $\mathbb{I\in}\mathcal{V}_{\perp}$,
implying $\mathrm{Tr}W_{ij,\,ab}^{(\alpha)}=0$, recovering again the partial
commutativity condition. 

A few comments in order: (i) $\dim\mathcal{V}$ is upper bounded by
the total number of matrices $W_{{i}j,\,{a}b}^{(\alpha)},\,M_{{i},\,{a}b}^{(\alpha)}$,
which is $s(s+1)r^2/2$, where we recall that $r$ is the rank of $\rho_{\bm{\lambda}}$ and $s$ is the number of the estimation parameter
However, it may be difficult to find simple,
yet universal rules that determines $\dim\mathcal{V}$ exactly, as
it can vary from case to case. For example, when $[L_{{i}},L_{j}]=0$,
$W_{{i}j,ab}=0$, which renders $\dim\mathcal{V}$ strictly smaller
than $s(s+1)r^2/2$. 
(ii) Nevertheless, by exploiting these structure,
some results can still be found on the linear independence or dependence
among the matrices$\,M_{i,\,ab}^{(\alpha)}$,  see Sec.~\ref{sec:lin-indept} in the Supplemental Material~\citep{SM}. 

Let us present central results that follows
from Eq.~(\ref{eq:Tr-condition}). 
\begin{obs} \label{obs:IC-POVM}
\textup{[No-go theorem for informationally-complete POVMs (IC-POVM)] 
For $s\ge2$ and non-singular QFIM, no IC-POVM can saturate the multiparameter
QCRB.}
\end{obs}

The informationally-complete (IC-POVM) or tomographically complete POVM is the rank-one POVM with $|\Omega|=d^{2}$. In this case, $\{E_{\omega}\}_{\omega=1}^{|\Omega|}$ spans
the real vector space of Hermitian operators denoted as $\mathrm{Herm}(d)$,
isomorphic to the Lie algebra $\mathfrak{u}(d)$.
For proof of above theorem, see the End Matter. For single parameter estimation where $s=1$, IC-POVM can saturate
the QFI only for states with particular property, see details in Theorem~\ref{thm:vanishing-M-condition}
in the Supplemental Material.
This observation is particularly relevant because meaningful quantum sensing tasks also involve non-singular QFIM with more than one estimation parameters.

Next we discuss the saturation of the QCRB, when PCC is satisfied. In this case, we know $\mathbb{I}\in\mathcal{V}_{\perp}$
and the remaining orthogonal basis must be traceless.
Thus, without
loss of generality, we can choose the basis to be 
\begin{equation}
T_{0}=\mathbb{I},\,\mathrm{Tr}T_{k}=0\,(k\ge1),\,\mathrm{Tr}(T_{k}T_{l})=d\delta_{kl}.\label{eq:su-alg-norm}
\end{equation}
Operationally, this basis can be found via the Gram-Schmidt orthogonalization
process under the trace inner product, starting from the set of matrices
$\{W_{{i}j,\,{a}b}^{(\alpha)},\,M_{{i},\,{a}b}^{(\alpha)}\}$.
We denote the linear subspace of the $\mathfrak{su}(d)$ Lie algebra $\mathcal{N}\equiv\mathrm{span}_{\mathbb{R}}\{\mathrm{i}T_{k}\}_{k\ge1}$
and the dimension of $\mathcal{N}$ as $n$. Clearly, $n=\dim\mathcal{V}_{\perp}-1\leq d^{2}-1$, where the upper bound can be reached if all the $W$ and $M$ matrices vanish.

\begin{obs} \label{obs:=000020dim-bound}
\textup{If $n<d-1$, i.e., $\dim\mathcal{V}_{\perp}<d$,
the QCRB cannot be saturated.}
\end{obs}
The proof is given in the End Matter. Geometrically, this means that the subspace $\mathcal{V}_{\perp}$ depicted in Fig.~\ref{fig:cen-geo}(b) is below the critical dimension for the rank-one POVM to be complete. Note that Observation~\ref{obs:=000020dim-bound} holds, independent of whether PCC holds or not. This implies that PCC is not sufficient for the saturation of the QCRB in general. 

\begin{thm}
\textup{
[Outcome-wise saturability criterion] 
\label{thm:outcome-wise=000020sat}
When PCC is satisfied, a rank-one projector $\Pi_{\omega}=\ket{\pi_{\omega}}\bra{\pi_{\omega}}$
saturates the outcome-wise QCRB, i.e., $F_{\omega}^{C}=F_{\omega}^{Q}$
if and only if
\begin{equation}
\Pi_{\omega}=\frac{1}{d}(\mathbb{I}-\bm{v}^{(\omega)}\cdot\bm{T}),\label{eq:opt-E-struct}
\end{equation}
where $\bm{v}^{(\omega)}$ is a $n$-dimensional non-zero real vector
such that the spectrum of the $d\times d$ traceless Hermitian matrix
$\bm{v}^{(\omega)}\cdot\bm{T}\equiv\sum_{k=1}^{n}v_{k}^{(\omega)}T_{k}$
are $(d-1)$-fold degenerate with degenerate eigenvalue $1$. }
\end{thm}
The proof is presented in the End Matter. The outcome-wise saturability for generic rank-one
POVM is trivial since it can be always renormalized to a rank-one
projector through $\Pi_{\omega}=E_{\omega}/\mathrm{Tr}E_{\omega}$.
Theorem~\ref{thm:outcome-wise=000020sat} immediately leads to an 
efficient numerical algorithm to search for optimal measurements, see the End Matter. Furthermore, by exploring the structure of the optimal rank-one projectors
giving (\ref{eq:opt-E-struct}), we can show more quantitatively when
PCC becomes sufficient:
\begin{thm}\label{thm:suff-ineq}
\textup{\label{thm:sufficiency} When $\dim\mathcal{V}_{\perp}$ is
large enough, more precisely, 
\begin{equation}
n\ge\max_{\mu\in[1,2\,\cdots d-2]}2\mu(d+1/2-\mu)+(d-2)\label{eq:suff-ineq}
\end{equation}
is satisfied, PCC is necessary and sufficient for the saturation of
the Helstrom bound. In particular, PCC ensures the existence of an
optimal rank-one \textit{projective} measurements, which can be constructed
iteratively.
}
\end{thm}
We prove by constructing the optimal measurements explicitly, see the End Matter and the Supplemental Material for schematica detailed proofs, respectively. A few comments in order: (a) We note that if an optimal rank-one
projective measurement exists, it can be always constructed via the
iterative procedure we presented in the proof of Theorem~\ref{thm:sufficiency}.
(b) It is expected that the lower bound on $n$ in inequality~(\ref{eq:suff-ineq})
above which PCC becomes sufficient may be further refined if the particular
structures or symmetries of $\rho_{\bm{\lambda}}$ and the SLD operators
are exploited. 

When $d\gg r,s$ and $n\gtrsim d^{2}$ while the maximum of r.h.s. of
(\ref{eq:suff-ineq}) scales $d^{2}/2$. Thus the inequality (\ref{eq:suff-ineq})
holds for large dimension of Hilbert space and we have the following observation:
\begin{obs}
\textup{[Asymptotic sufficiency of PCC] When $d\gg r,\,s$, i.e.,
the dimension of the Hilbert space is much larger than the rank of the density
operator and the number of estimation parameters, PCC becomes necessary
and sufficient for the saturation of the QCRB.
In particular, if, for continuous-variable systems, $\rho_{\bm{\lambda}}$ is of finite rank and contains a finite number of estimation parameters, then a finite-dimensional $d$-Lie algebra which contains $\mathcal{V}$ as a subspace can be always constructed with arbitrary large $d$~\footnote{More precisely, we consider the $\mathfrak{su}(d)$ Lie algebra is defined over on the $d$-dimensional Hilbert space containing the subspace $\mathrm{span}_\mathbb{C}\{\ket{\psi_{\bm{a\lambda}}},\ket{\partial_i \psi_{\bm{a\lambda}}}\}_{a,\,i}$.}. 
As such, PCC is always
necessary and sufficient for estimation of finite-rank density matrix
with finite number of parameters.}
\end{obs}

\textit{{Example: Quasi-pure States.}}\textit{---} One might expect that quantum states satisfying PCC may be complicated and not general enough to be of metrological interest. However, it was shown in the context of postselected quantum metrology ~\cite{yang2024purestate,yang2024theoryof} that for quasi-pure states, which behaves very much like pure states, PCC becomes simple and resembles the one for pure states. In particular, we consider the following class of quasi-pure states generated via the classical correlations between a primary system and an ancilla:
\begin{equation}\label{eq:QP-bipartite}
    \rho_{\bm{\lambda}}=(U_{\bm{\lambda}}\otimes\mathbb{I})\rho_{0}(U_{\bm{\lambda}}^{\dagger}\otimes\mathbb{I})=\sum_{a=1}^{r}q_{a}\ket{\phi_{a\bm{\lambda}}}\bra{\phi_{a\bm{\lambda}}}\otimes|a\rangle\langle a|,
\end{equation}
where $\rho_{0}=\sum_{a=1}^{r}q_{a}|\phi_{a}\rangle\langle\phi_{a}|\otimes|a\rangle\langle a|$, $\ket{\phi_{a\bm{\lambda}}}\equiv U_{\bm{\lambda}}\ket{\phi_a}$, $\ket{a}$ are orthonormal while $\ket{\phi_a}$ is normalized but not necessarily orthogonal to each other.
Remarkably, it was shown in Ref.~\citep{yang2024purestate} that the PCC for states given by Eq.~\eqref{eq:QP-bipartite} reduces to 
\begin{equation}
\mathrm{Im}\braket{D_{i}\phi_{a\bm{\lambda}}|D_{j}\phi_{a\bm{\lambda}}}=0,\,\forall i,\,j,\,a,\label{eq:QP-PCC}
\end{equation}
where $\ket{D_{i}\psi_{a\bm{\lambda}}}$ is the covariant derivative
defined as $\ket{D_{i}\phi_{a\bm{\lambda}}}\equiv (\mathbb{I}-\ket{\phi_{a\bm{\lambda}}}\bra{\phi_{a\bm{\lambda}}})\ket{\partial_{i}\phi_{a\bm{\lambda}}}$. It can be readily checked that, unless for full rank states, the PCC also holds
while the SLD does not mutually commute in general~\cite{yang2024purestate}. 
Thanks to the Hollowization theorem~\ref{thm:Simultaneous=000020Hollowization=000020Theorem}, we can show that
\textup{If the PCC~\eqref{eq:QP-PCC} is satisfied, the measurement incompatibility between different eigenstates disappears. Furthermore, the QCRB for the state given by Eq.~\eqref{eq:QP-bipartite} can be saturated by the local measurements with classical communications (LMCC) $\{\ket{\pi_{\nu,a}}=\ket{e^{(a)}_{\nu}}\otimes \ket{a}\}$, where $\{\ket{e^{(a)}_{\nu}}\}$ is the optimal measurement basis for the saturating the QCRB of each of the pure state $\{\ket{\phi_{a\bm{\lambda}}}\}$, which can be constructed systemmatically Refs.~\citep{matsumoto2002anew} and ~\citep{pezz`e2017optimal}.}
Furthermore, if $U_{\bm{\lambda}}$ in Eq.~\eqref{eq:QP-bipartite} is generated by commuting Hamiltonian, i.e., $U_{\bm{\lambda}}=\mathrm{e}^{-\mathrm{i}\sum_{j}\lambda_{j}H_{j}}$ with $[H_{j},\,H_{k}]=0,\forall j,\,k$, the QCRB for the state is always saturable for all initial state $\rho_0$.
The proofs of these results, and their applications to a magnetometry with two-qubit primary system and a qubit ancilla can be found in the End Matter. It is worth noting that for generic quasi-pure states, whether PCC is sufficient or not is still open, as conjectured in Ref.~\cite{yang2024purestate}.

\textit{{Discussion and Conclusion.}}\textit{---}
We solve the open question on single-copy optimal measurement in multi-parameter
quantum metrology when PCC holds. There could be several ramifications of our results.
First, even when PCC fails, our algorithm can still be used to construct rank-one measurement bases that saturate the QCRB outcome-wise, although they may not satisfy the POVM completeness (closure) relation.
This observation suggests a systematic route to practical approximations of the QCRB, which may be of broad interest for applications.
Second, our work motivates a more comprehensive mathematical study of simultaneous hollowization for sets of traceless matrices, a topic that appears largely unexplored in the current literature.
Progress along these lines would sharpen our understanding of measurement incompatibility in multiparameter sensing and, more broadly, of the geometric constraints imposed by quantum theory.

\textit{{Acknowledgement.}}\textit{---} We thank Wenquan Liu, Yang Wu, Liang Xu for discussions on the experimental relevance of quasi-pure state estimation. J.Y. acknowledges support from Zhejiang University start-up grants, Zhejiang Key Laboratory of R\&D and Application of Cutting-edge Scientific Instruments, and Wallenberg Initiative on Networks and Quantum Information (WINQ). S.I. acknowledges support from JST ASPIRE (JPMJAP2339). L.P. acknowledges support from the QuantERA project SQUEIS (Squeezing enhanced inertial sensing), funded by the European Union's Horizon Europe Program and the Agence Nationale de la Recherche (ANR-22-QUA2-0006).
This publication has received funding under Horizon Europe programme HORIZON-CL4-2022-QUANTUM-02-SGA via the project 10113690 PASQuanS2.1.

\let\oldaddcontentsline\addcontentsline     
\renewcommand{\addcontentsline}[3]{}         

\bibliography{Multi-Parameter-Metrology}

\clearpage\newpage{}
\begin{center}
{\large\textbf{End Matter}}{\large\par}
\par\end{center}
\section{Proof of Observation~\ref{obs:IC-POVM}}
\begin{proof}
We prove by contradiction. We note the fact that IC POVM forms a complete basis on $\mathrm{Herm}(d)$. According to Theorem~\ref{thm:Simultaneous=000020Hollowization=000020Theorem}, we have $M_{i,\,ab}=0$ for all
$i\in[1,s],\,a,b\in[1,\,r]$. Then according to Theorem~\ref{thm:vanishing-M-condition}
in the Supplemental Material, we know $L_{i}=\alpha_{i}\Pi_{\mathrm{s}}$, where $\Pi_{\mathrm{s}}$ is the project to the support space of $\rho_{\bm{\lambda}}$ and
$\alpha_{i}\in\mathbb{R}$. This is consistent with the condition
that $W_{ij,\,ab}=0$, for all $i,j\in[1,s],\,a,b\in[1,\,r]$. However,
in this case, the QFIM element reads $F_{jk}^{Q}=\alpha_{j}\alpha_{k}$,
which is of rank-one and thus non-invertible, which is in contradiction with the premise that the QFIM is non-singular.
\end{proof}

\section{Proof of Observation~\ref{obs:=000020dim-bound}}
\begin{proof}
Again we prove by contradiction. We first assume there exists one optimal measurements consisting of rank-one measurement operators, which is denoted as $\{E_{\omega}\}_{\omega\in\Omega}$.  Since $|\Omega|\ge d$ and the optimal POVM rank-one operator 
forms a linearly independent basis on $\mathcal{V}_{\perp}$, therefore
$\dim\mathcal{V}_{\perp}$ must be greater or equal to $d$, which in contradiction to condition that $\dim \mathcal{V}_{\perp}<d$.
\end{proof}

\section{Proof of Theorem~\ref{thm:outcome-wise=000020sat}}

\label{sec:Proof-Outcome-wise-saturability}
\begin{proof}
We first prove the forward direction: if a projector $\Pi_\omega$ saturate the outcome-wise QCRB, then according the geometric interpretation of the saturation condition~\eqref{eq:Tr-condition}, we know
\begin{equation}
\frac{1}{d}(\alpha\mathbb{I}-\bm{v}^{(\omega)}\cdot\bm{T})=\ket{\pi_{\omega}}\bra{\pi_{\omega}},
\end{equation}
where $\alpha\in\mathbb{R}$ and $\bm{v}^{(\omega)}$ is a real vector.
This would immediate imply that 
\begin{equation}
\bm{v}^{(\omega)}\cdot\bm{T}=\alpha\mathbb{I}-d\ket{\pi_{\omega}}\bra{\pi_{\omega}}.
\end{equation}
The traceless property of $T_{k}$ implies that $\alpha=1$. Then
spectrum on the r.h.s. is $\{1-d,1,1,\cdots1\}$, which completes the proof.
The backward direction is obvious: assuming Eq.~\eqref{eq:opt-E-struct} holds, it apparent satisfies Eq.~\eqref{eq:Tr-condition}.
\end{proof}

\section{Numerical algorithms for searching the optimal measurements}
\begin{enumerate}
\item Choose a set of orthonormal complete traceless basis $\{T_{k}\}_{k=1}^{n}$
on $\mathcal{V}_{\perp}$ that satisfies (\ref{eq:su-alg-norm})
and impose $d-1$ eigenvalues of $\bm{v}\cdot\bm{T}$ are equal to $1$,
where $\bm{v}$ is $n$-dimensional real vector. This amounts to imposing
$d-1$ real constraints, which is possible only for $n\ge d-1$, consistent
with the Observation~\ref{obs:=000020dim-bound}. The number of distinct
ways of imposing such a$(d-1)$ degeneracy is $d$, which leads to
a set of vectors and optimal rank-one projectors denoted as $\{\bm{v}^{(\omega)},\,\Pi_{\omega}\}$.
\item Search for non-negative $\alpha^{(\omega)}\in[0,1]$ that satisfies
$\sum_{\omega=1}^{|\Omega|}\alpha^{(\omega)}=d,\,\sum_{\omega=1}^{|\Omega|}\alpha^{(\omega)}\bm{v}^{(\omega)}=0$.
These two conditions guarantees that to the set of optimal rank-one
POVM $E_{\omega}\equiv\alpha^{(\omega)}\Pi_{\omega}$ is complete,
i.e., $\sum_{\omega=1}^{|\Omega|}\alpha^{(\omega)}\Pi_{\omega}=\mathbb{I}$, see Sec.~\ref{sec:justification-efficient-search-algorithm}
in~\citep{SM} for details. If such $\{\alpha^{(\omega)}\}$ does
not exists, then the QCRB is not saturable. In particular,
to search for projective measurements, $|\Omega|=d$ and $\alpha^{(\omega)}=1$
and one just need to search for $d$ vectors such that $\sum_{\omega=1}^{|\Omega|}\bm{v}^{(\omega)}=0$. 
\end{enumerate}

\section{Schematic Proof of Theorem~\ref{thm:suff-ineq}}

\begin{proof}
We prove by giving a recipe of constructing the optimal rank-one projective
measurements. A comprehensive description of the procedure can be
found in Sec.~\ref{sec:Detailed-proof-sufficiency} in~\citep{SM}.
Here we briefly mention the idea behind the iterative construction.
Clearly, the first projector can be constructed via the algorithm
above. Assuming first $\mu$ rank-one projectors, denoted as $\{\Pi_{q}\}_{q=1}^{\mu}$,
has been constructed, the $(\mu+1)$-th can constructed as follows:
\begin{enumerate}
\item We denote the orthogonal complement of the linear subspace $\mathrm{span}_{\mathbb{R}}\{\Pi_{q}\}_{q=1}^{\mu}$
in $\mathcal{V}_{\perp}$ as $\mathcal{V}_{\perp}^{(\mu+1)}$, with
$\dim\mathcal{V}_{\perp}^{(\mu+1)}=n-\mu+1$. Clearly $\mathbb{I}^{(\mu+1)}\equiv\mathbb{I}-\sum_{q=1}^{\mu}\Pi_{q}\in\mathcal{V}_{\perp}^{(\mu+1)}$.
All other basis on $\mathcal{V}_{\perp}^{(\mu+1)}$, denoted as $T_{k}^{(\mu+1)}$
with $k=1,2,\cdots n-\mu$, that are orthogonal to $\mathbb{I}^{(\mu+1)}$
must be traceless and satisfies 
\[
\langle T_{k}^{(\mu+1)},\Pi_{q}\rangle=\braket{\pi_{q}|T_{k}^{(\mu+1)}|\pi_{q}}=0,\,q\in[1,\,\mu].
\]
 We choose the normalization 
\begin{equation}
\langle T_{k}^{(\mu+1)},\,T_{l}^{(\mu+1)}\rangle=(d-\mu)\delta_{kl}.
\end{equation}
\item Clearly $\Pi_{\mu+1}$ lies in $\mathcal{V}_{\perp}^{(\mu+1)}$. Using
a theorem that is analogous to Theorem~\ref{thm:outcome-wise=000020sat}
(see Sec.~\ref{sec:justification-efficient-search-algorithm} in~\citep{SM}),
the only possibility to form a rank-one project is that there exists
a real vector $\bm{v}^{(\mu+1)}$ such that (i) The kernel space of$\bm{v}^{(\mu+1)}\cdot\bm{T}^{(\mu+1)}\equiv\sum_{k=1}^{n-\mu}v_{k}^{(\mu+1)}T_{k}$
is equal to $\mathrm{span}_{\mathbb{C}}\{\ket{\pi_{q}}\}_{q=1}^{\mu}$
(ii) The $(d-\mu-1)$ of the non-zero eigenvalues must be $1$, with
the remaining one has to be $-(d-\mu-1)$ due to the traceless property
of $T_{k}^{(\mu+1)}$.
\end{enumerate}
Such an iterative procedure can be carried out until the construction
of the $(d-1)$th rank-one projector, corresponding to $\mu=d-2$.
It turns out that for constructing the $(\mu+1)$-th projector, condition
(i) \textit{at most} $2\mu(d-\mu)+\mu-1$ real constraints while condition
(ii) imposes \textit{at most} $(d-\mu-1)$ constraints. On the other
hand, we note the number of free parameters is $n-\mu$. Therefore,
$(\mu+1)$-th optimal rank-one projector can be constructed successfully
if $n-\mu\ge2\mu(d-\mu)+d-2$. Imposing such an inequality holds for
all $\mu\in[1,\,d-2]$ leads to the inequality~\eqref{eq:suff-ineq}.
\end{proof}

\section{Proof of the results on quasi-pure states} 

\begin{proof}
For the state given by Eq.~\eqref{eq:QP-bipartite}, the SLD can be expressed as~\citep{yang2024purestate}
\begin{equation}
L_{i}=\sum_{a=1}^{r}L_{i,\,a}\otimes|a\rangle\bra{a},
\end{equation}
where $L_{i,\,a}$ is the SLD for the pure state $|\phi_{a\bm{\lambda}}\rangle$
defined as 
\begin{equation}
L_{i,a}=2(\ket{D_{i}\phi_{a\bm{\lambda}}}\bra{\phi_{a\bm{\lambda}}}+\ket{\phi_{a\bm{\lambda}}}\bra{D_{i}\phi_{a\bm{\lambda}}}).
\end{equation}
Therefore, it can be readily calculated
\begin{align}
W_{ij,\,ab} & =\mathcal{W}_{ij,\,ab}\otimes\ket{a}\bra{b},\\
M_{i,\,ab} & =\mathcal{M}_{i,\,ab}\otimes|a\rangle\langle b|,
\end{align}
where 
\begin{align}
\mathcal{W}_{ij,\,ab} & \equiv L_{i,\,a}\ket{\phi_{a\bm{\lambda}}}\bra{\phi_{b\bm{\lambda}}}L_{j,\,b}-L_{j,\,a}\ket{\phi_{a\bm{\lambda}}}\bra{\phi_{b\bm{\lambda}}}L_{i,\,b}\nonumber \\
 & =2\left(\ket{D_{i}\phi_{a\bm{\lambda}}}\bra{D_{j}\phi_{b\bm{\lambda}}}-\ket{D_{j}\phi_{a\bm{\lambda}}}\bra{D_{i}\phi_{b\bm{\lambda}}}\right),
\end{align}
and 
\begin{align}
\mathcal{M}_{i,\,ab} & \equiv L_{i,\,a}\ket{\phi_{a\bm{\lambda}}}\bra{\phi_{b\bm{\lambda}}}-\ket{\phi_{a\bm{\lambda}}}\bra{\phi_{b\bm{\lambda}}}L_{i,\,b}\nonumber \\
 & =2\left(\ket{D_{i}\phi_{a\bm{\lambda}}}\bra{\phi_{b\bm{\lambda}}}-\ket{\phi_{a\bm{\lambda}}}\bra{D_{i}\phi_{b\bm{\lambda}}}\right).
\end{align}
Consider the local measurement with classical communications (LMCC) with the measurement basis of the form $\ket{\pi_{\nu,\,a}}=\ket{e^{(a)}_{\nu}}\otimes\ket{a}$.
Note that
\begin{align}
\braket{\pi_{\nu,\,a}|W_{ij,\,cd}|\pi_{\nu,\,a}} & =\braket{e^{(a)}_{\nu}|\mathcal{W}_{ij,\,cd}|e^{(a)}_{\nu}}\delta_{ac}\delta_{ad},\\
\braket{\pi_{\nu,\, a}|M_{i,\,cd}|\pi_{\nu,\,a}} & =\braket{e^{(a)}_{\nu}|\mathcal{M}_{i,\,cd}|e^{(a)}_{\nu}}\delta_{ac}\delta_{ad}.
\end{align}
Then according to the Hollowization theorem~\ref{thm:Simultaneous=000020Hollowization=000020Theorem}, this set of measurements
optimal if and only 
\begin{equation}
\braket{e^{(a)}_{\nu}|\mathcal{W}_{ij,\,aa}|e^{(a)}_{\nu}}=0,\,\braket{e^{(a)}_{\nu}|\mathcal{M}_{i,\,aa}|e^{(a)}_{\nu}}=0.\label{eq:Hollowization-phi}
\end{equation}
Applying the Hollowization theorem again to find the optimal the optimal
measurement condition for the pure state $\ket{\phi_{a\bm{\lambda}}}$,
one immediately observe that it is the same as Eq.~(\ref{eq:Hollowization-phi}).
Since $\ket{\phi_{a\bm{\lambda}}}$ satisfies the WCC, the optimal
measurement $\{\ket{e_{\nu_a}}\}$ can be constructed systematically
from the recipes proposed by Mastumoto and Pezze et al~\citep{matsumoto2002anew,pezz`e2017optimal}. It is clear that LMCC constructed in this way complete, since
\begin{equation}
   \sum_a \sum_{\nu} \ket{e^{(a)}_{\nu}}\bra{e^{(a)}_{\nu}}\otimes \ket{a}\bra{a}=\sum_a \ket{a}\bra{a}=\mathbb{I}.
\end{equation}

\end{proof}

\section{Application to systems consisting of a two-qubit primary system and qubit ancilla}

We consider a two-parameter estimation problem with the primary system
being two qubits, where $H_{1}=\lambda_{1}\sigma_{z}^{(1)}\sigma_{z}^{(2)}$
and $H_{2}=\lambda_{2}\sigma_{x}^{(1)}\sigma_{x}^{(2)}$. Clearly
here $[H_{1},\,H_{2}]=0$. We consider the initial state 
\begin{equation}
\rho_{0}=q\ket{0,+}\bra{0,+}\otimes\ket{0}\bra{0}+(1-q)\ket{1,\varphi}\bra{1,\varphi}\otimes\ket{1}\bra{1},
\end{equation}
where $\ket{\varphi}=\cos(\theta/2)\ket{0}+\sin(\theta/2)\ket{1}$.
It can be calculated that 
\begin{equation}
\rho_{\bm{\lambda}}=q\ket{\phi_{0\bm{\lambda}}}\bra{\phi_{0\bm{\lambda}}}\otimes\ket{0}\bra{0}+(1-q)\ket{\phi_{1\bm{\lambda}}}\bra{\phi_{1\bm{\lambda}}}\otimes\ket{1}\bra{1},
\end{equation}
where $\ket{\phi_{0\bm{\lambda}}}=U_{\bm{\lambda}}\ket{0,+}$ and
$\ket{\phi_{1\bm{\lambda}}}=U_{\bm{\lambda}}\ket{1,\varphi}$.
Since this state is a quasi-pure states~\cite{yang2024purestate}, the QFIM can be calculated
as 
\begin{align}
F^{Q} & =qF^{Q}[\ket{\phi_{0\bm{\lambda}}}]+(1-q)F^{Q}[\ket{\phi_{1\bm{\lambda}}}]=\begin{bmatrix}4 & 0\\
0 & 4q+4(1-q)\sin^{2}\theta
\end{bmatrix}.
\end{align}
Following Matsumoto and Pezze~\cite{matsumoto2002anew,pezz`e2017optimal}, we first construct the regular optimal
measurements for the estimation of $\ket{\phi_{0\bm{\lambda}}}$ and
$\ket{\phi_{1\bm{\lambda}}}$ locally $\bm{\lambda}=0$. Optimal projectors
for estimating $\ket{\phi_{0\bm{\lambda}}}$ is
\begin{align}
\ket{e_{1}^{(0)}} & =\frac{1}{\sqrt{3}}\left(\ket{0+}+\mathrm{i}\sqrt{2}\ket{1+}\right),\\
\ket{e_{2}^{(0)}} & =\frac{1}{\sqrt{3}}\left(\ket{0+}-\mathrm{i}\frac{1}{\sqrt{2}}\ket{1+}+\mathrm{i}\sqrt{\frac{3}{2}}\ket{0-}\right),\\
\ket{e_{3}^{(0)}} & =\frac{1}{\sqrt{3}}\left(\ket{0+}-\mathrm{i}\frac{1}{\sqrt{2}}\ket{1+}-\mathrm{i}\sqrt{\frac{3}{2}}\ket{0-}\right),\\
\ket{e_{4}^{(0)}} & =\ket{1-},
\end{align}
while the optimal regular projectors for estimating $\ket{\phi_{1\bm{\lambda}}}$
can be constructed as follows:
\begin{align}
\ket{e_{1}^{(1)}} & =\frac{1}{\sqrt{3}}\left(\ket{1\varphi}+\mathrm{i}\sqrt{2}\ket{0\varphi_{X}}\right),\\
\ket{e_{2}^{(1)}} & =\frac{1}{\sqrt{3}}\left(\ket{1\varphi}-\mathrm{i}\frac{1}{\sqrt{2}}\ket{0\varphi_{X}}+\mathrm{i}\sqrt{\frac{3}{2}}\ket{1\varphi^{\perp}}\right),\\
\ket{e_{3}^{(1)}} & =\frac{1}{\sqrt{3}}\left(\ket{1\varphi}-\mathrm{i}\frac{1}{\sqrt{2}}\ket{0\varphi_{X}}-\mathrm{i}\sqrt{\frac{3}{2}}\ket{1\varphi^{\perp}}\right),\\
\ket{e_{4}^{(1)}} & =\ket{0\varphi_{X}^{\perp}},
\end{align}
where 
\begin{align}
\ket{\varphi^{\perp}} & =\sin(\theta/2)\ket{0}-\cos(\theta/2)\ket{1},\\
\ket{\varphi_{X}} & \equiv\sigma_{x}\ket{\varphi}=\cos(\theta/2)\ket{1}+\sin(\theta/2)\ket{0},\\
\ket{\varphi_{X}^{\perp}} & \equiv\sigma_{x}\ket{\varphi^{\perp}}=\sin(\theta/2)\ket{1}-\cos(\theta/2)\ket{0}.
\end{align}
Following theorem in the main text, we consider the measurements 
$\{\ket{\pi_{\nu,\,a}}\equiv\ket{e^{(a)}_{\nu}}\otimes\ket{a}\}$,
which is optimal. Indeed, it can be easily verified that the
CFIM $F^{C}$ is always identical to $F^{Q}$, although this choice
of measurement is meant to be locally optimal at $\bm{\lambda}=0$.

\clearpage\newpage{}
\setcounter{equation}{0} 
\setcounter{section}{0} 
\setcounter{subsection}{0} 
\setcounter{thm}{0} 
\setcounter{lem}{0} 
\renewcommand{\thesection}{S\arabic{section}}
\renewcommand{\theequation}{S\arabic{equation}}
\renewcommand{\thethm}{S\arabic{thm}}
\renewcommand{\theobs}{S\arabic{obs}}
\renewcommand{\thelem}{S\arabic{lem}}
\onecolumngrid 
\setcounter{enumiv}{0}
\begin{center}
{\large\textbf{Supplemental Materials for "A geometric criterion for optimal measurements in multiparameter quantum metrology"}}{\large\par}
\par\end{center}

\let\addcontentsline\oldaddcontentsline  

\addtocontents{toc}{\protect\thispagestyle{empty}}
\pagenumbering{gobble}

\tableofcontents{}

\section{Proof of Theorem~\ref{thm:Simultaneous=000020Hollowization=000020Theorem}
in the main text}

\label{sec:Proof-Hollowization-Theorem}

\subsection{Null POVM basis}

Our goal is to eliminate the proportionality constant $\eta_{\omega,\,ij}$
so that Eq.~(\ref{eq:Opt-Null}) can be recast into a compact form.
To this end, we show that for a given pair $(i,j)$, Eq.~(\ref{eq:Opt-Null})
is equivalent to 
\begin{equation}
\frac{\braket{\psi_{b\bm{\lambda}}\big|L_{i}\big|\pi_{\omega}}}{\braket{\psi_{b\bm{\lambda}}\big|L_{j}\big|\pi_{\omega}}}=\frac{\braket{\pi_{\omega}\big|L_{i}\big|\psi_{a\bm{\lambda}}}}{\braket{\pi_{\omega}\big|L_{j}\big|\psi_{a\bm{\lambda}}}},\,\forall a,\,b.\label{eq:Opt-Null-ratio}
\end{equation}
Following from Eq.~(\ref{eq:Opt-Null-ratio}), we find 
\begin{equation}
\braket{\pi_{\omega}\big|W_{ij,\,ab}\big|\pi_{\omega}}=0,\,\forall i,\,j,\,a,\,b,
\end{equation}
where 
\begin{equation}
W_{ij,\,ab}\equiv L_{i}\ket{\psi_{a\bm{\lambda}}}\bra{\psi_{b\bm{\lambda}}}L_{j}-L_{j}\ket{\psi_{a\bm{\lambda}}}\bra{\psi_{b\bm{\lambda}}}L_{i}.
\end{equation}
To see the reverse direction also holds, let us first note when $a=b$,
Eq.~(\ref{eq:Opt-Null-ratio}) indicates that 
\begin{equation}
\eta_{\omega,\,ij}^{(a)}\equiv\frac{\braket{\psi_{a\bm{\lambda}}\big|L_{i}\big|\pi_{\omega}}}{\braket{\psi_{a\bm{\lambda}}\big|L_{j}\big|\pi_{\omega}}}
\end{equation}
is a real number. Then for $a\neq b$, Eq.~(\ref{eq:Opt-Null-ratio})
clearly indicates that $\eta_{\omega,\,ij}^{(a)}$ is independent
of $a$.

\subsection{Regular POVM basis}

For given $i$ and $\omega$, Eq.~(\ref{eq:Opt-Reg}) is equivalent
to 
\begin{equation}
\frac{\braket{\psi_{b\bm{\lambda}}\big|L_{i}\big|\pi_{\omega}}}{\braket{\psi_{b\bm{\lambda}}\big|\pi_{\omega}}}=\frac{\braket{\pi_{\omega}\big|L_{i}\big|\psi_{a\bm{\lambda}}}}{\braket{\pi_{\omega}\big|\psi_{a\bm{\lambda}}}},\,\forall a,\,b.\label{eq:Opt-Reg-ratio}
\end{equation}
Follows from Eq.~(\ref{eq:Opt-Reg}) straightforwardly, we obtain
\begin{equation}
\braket{\pi_{\omega}\big|M_{i,\,ab}\big|\pi_{\omega}}=0,
\end{equation}
where 
\begin{equation}
M_{i,\,ab}\equiv L_{i}\ket{\psi_{a\bm{\lambda}}}\bra{\psi_{b\bm{\lambda}}}-\ket{\psi_{a\bm{\lambda}}}\bra{\psi_{b\bm{\lambda}}}L_{i}.
\end{equation}
 For the reverse direction, when $a=b$, Eq.~(\ref{eq:Opt-Reg-ratio})
implies that 
\begin{equation}
\xi_{\omega,\,i}^{(a)}\equiv\frac{\braket{\psi_{a\bm{\lambda}}\big|L_{i}\big|\pi_{\omega}}}{\braket{\psi_{a\bm{\lambda}}\big|\pi_{\omega}}},
\end{equation}
is a real number. For $a\neq b$, Eq.~(\ref{eq:Opt-Reg-ratio}) clearly
indicates that $\xi_{\omega,\,i}^{(a)}$ is independent of $a$.

Now we take one step further. Note that taking two equations from
Eq.~(\ref{eq:Opt-Reg-ratio}) for different $i$ and $j$ respectively
and then dividing them, we obtain Eq.~(\ref{eq:Opt-Null-ratio}).
However, the reverse direction is not guaranteed. To elaborate, we
rewrite Eq.~(\ref{eq:Opt-Null-ratio}) as 

\begin{equation}
\frac{\braket{\psi_{b\bm{\lambda}}\big|L_{i}\big|\pi_{\omega}}/\braket{\psi_{b\bm{\lambda}}\big|\pi_{\omega}}}{\braket{\psi_{b\bm{\lambda}}\big|L_{j}\big|\pi_{\omega}}/\braket{\psi_{b\bm{\lambda}}\big|\pi_{\omega}}}=\frac{\braket{\pi_{\omega}\big|L_{i}\big|\psi_{a\bm{\lambda}}}/\braket{\pi_{\omega}\big|\psi_{a\bm{\lambda}}}}{\braket{\pi_{\omega}\big|L_{j}\big|\psi_{a\bm{\lambda}}}/\braket{\pi_{\omega}\big|\psi_{a\bm{\lambda}}}}, \forall a,b,
\end{equation}
which implies 
\begin{equation}
\eta_{\omega,\,ij}^{(b)}=\frac{\xi_{\omega,\,i,}^{(b)}}{\xi_{\omega,\,j}^{(b)}}=\left(\frac{\xi_{\omega,\,i}^{(a)}}{\xi_{\omega,\,j}^{(a)}}\right)^{*}=\left(\eta_{\omega,\,ij}^{(a)}\right)^{*}.
\end{equation}
As before, from this equation, we conclude that $\eta_{\omega,\,ij}^{(a)}$
is real and independent of $a$ and
\begin{equation}
\frac{\xi_{\omega,\,i}^{(a)}}{\xi_{\omega,\,j}^{(a)}}=\eta_{\omega,ij}\in\mathbb{R}.
\end{equation}
This means if when Eq.~(\ref{eq:Opt-Null-ratio}) holds for regular
operators, for a given pair $(i,\,j)$, it only ensures the ratios
between $\xi_{\omega,\,i}^{(a)}$'s are real and independent of $a$,
not $\xi_{\omega,\,i}^{(a)}$ themselves. As such, to guarantee the optimality of a regular measurement operator,
in addition to imposing Eq.~(\ref{eq:Opt-Null-ratio}), we also require
$\xi_{\omega,\,i}^{(a)}$ is real and independent of $a$, which means that Eq.~(\ref{eq:Opt-Reg-ratio}) must hold. Note that when if $E_{\omega}$
is null, Eq.~(\ref{eq:Opt-Reg-ratio}) holds automatically. 

To summarize, the optimal condition becomes the following: For measurement
operator $E_{\omega}$ of any type, either null or regular, we require
Eq.~~(\ref{eq:Opt-Null-ratio}) and Eq.~(\ref{eq:Opt-Reg-ratio}) hold.   On the other hand,  Eq.~(\ref{eq:Opt-Null-ratio}) is equivalent to Eq.~(\ref{eq:W-opt})
in Theorem~\ref{thm:Simultaneous=000020Hollowization=000020Theorem}
while Eq.~(\ref{eq:Opt-Reg-ratio}) 
is equivalent to Eq.~(\ref{eq:M-opt}) in Theorem~\ref{thm:Simultaneous=000020Hollowization=000020Theorem}.
Note that given Eq.~(\ref{eq:Opt-Reg-ratio})  and  Eq.~(\ref{eq:M-opt}), there is no need to distinguish whether the measurement operator
is null or regular.  This concludes the proof of the first part of  Theorem~\ref{thm:Simultaneous=000020Hollowization=000020Theorem}
in the main text.

\section{Proof of Eq.~(\ref{eq:W-opt-pair}) and Eq.~(\ref{eq:M-opt-pair}) in the main
text}

\label{sec:Proof-Pair-Hollowization-Theorem}

Before we prove Eq.~\eqref{eq:W-opt-pair} and Eq.~\eqref{eq:M-opt-pair}, we present a useful observation as follows:

\begin{obs}\label{obs:non-singular}
\textup{If $F^Q_{\omega}\neq 0$, i.e., $[F^Q_{\omega}]_{ij}=\mathrm{Tr}(\rho_{\bm{\lambda}}L_iE_{\omega}L_j)\neq0$, there must exists an estimation parameter $\lambda_{\bar{i}}$ and an eigenstate $\ket{\psi_{\bar{a}\bm{\lambda}}}$, called central estimation parameter and central eigenstate respectively, such that 
\begin{equation}\label{eq:non-singular}
\braket{\psi_{\bar{a}\bm{\lambda}}|L_{\bar{i}}|\pi_{\omega}}\neq0.
\end{equation}
}
\end{obs}
\begin{proof}
We prove by constradictions. Assuming $\braket{\psi_{{a}\bm{\lambda}}|L_{{i}}|\pi_{\omega}}=0$  for all $a$ and $i$. Then it can be readily found that 
\begin{equation}
    [F^Q_{\omega}]_{ij}=\sum_a p_a \mathrm{Re}[\braket{\psi_{{a}\bm{\lambda}}|L_{{i}}|\pi_{\omega}}\braket{\pi_{\omega}|L_{{j}}|\psi_{{a}\bm{\lambda}}}]=0,
\end{equation}
which is in contradiction to the premise that $F^Q_{\omega}\neq 0$.
\end{proof}

\subsection{Choose a central estimation parameter}

By definition, we know 
\begin{equation}\label{eq:xi-constraint}
\xi_{\omega,\,j}^{(b)}=\xi^{(b)}_{\omega,\,\bar{i}}\eta^{(b)}_{\omega,j\bar{i}}.
\end{equation}
Let us assume Eq.~\eqref{eq:W-opt-pair} and Eq.~\eqref{eq:M-opt-pair} hold for the
estimation parameter $\lambda_{\bar{i}}$ and the central eigenstate $\ket{\psi_{\bar{a}\bm{\lambda}}}$.  Then we know that $\eta^{(b)}_{\omega,j\bar{i}}$ and $\xi^{(b)}_{\omega,\,\bar{i}}$ are
real and independent of $b$. Therefore we conclude that $\xi_{\omega,\,j}^{(b)}$ must
be also real and independent of $b$, which implies that Eq.~\eqref{eq:M-opt} holds for all the estimation parameters and eigenstates.

On the other hand, we should avoid $\xi^{(b)}_{\omega,\,\bar{i}}$ and $\eta^{(b)}_{\omega,j\bar{i}}$ being indetermine, i.e. of the type $0/0$. In this case,  the constraint given by Eq.~\eqref{eq:xi-constraint}  is trivial in that it is not able to constrain  $\xi_{\omega,\,j}^{(b)}$ such that is it is real and independent of $b$. Fortunately, thanks to Observation~\ref{obs:non-singular},when Eq.~\eqref{eq:M-opt-pair} holds, we know $\xi_{\omega,\,j}^{(b)}=\xi_{\omega,\,j}^{(\bar{a})}$ is not of the type $0/0$.
Similar arugements can be applied to conclude that $\eta^{(b)}_{\omega,j\bar{i}}$ is not of the type $0/0$ when Eq.~\eqref{eq:W-opt-pair} and  hold.

\subsection{Choose a central eigenstate}
Given Eq.~\eqref{eq:W-opt-pair} and Eq.~\eqref{eq:M-opt-pair} hold, it can be readily shown that for any given pair $(i,j)$ and $(a,b)$,
\begin{align}
\frac{\braket{\psi_{b\bm{\lambda}}\big|L_{i}\big|\pi_{\omega}}}{\braket{\psi_{b\bm{\lambda}}\big|L_{j}\big|\pi_{\omega}}} & =\frac{\braket{\psi_{b\bm{\lambda}}\big|L_{i}\big|\pi_{\omega}}}{\braket{\psi_{b\bm{\lambda}}\big|L_{\bar{i}}\big|\pi_{\omega}}}\frac{\braket{\psi_{b\bm{\lambda}}\big|{L_{\bar{i}}}\big|\pi_{\omega}}}{\braket{\psi_{b\bm{\lambda}}\big|L_{j}\big|\pi_{\omega}}}=\frac{\braket{\pi_{\omega}\big|L_{i}\big|\psi_{\bar{a}\bm{\lambda}}}}{\braket{\pi_{\omega}\big|L_{\bar{i}}\big|\psi_{\bar{a}\bm{\lambda}}}}\frac{\braket{\pi_{\omega}\big|L_{\bar{i}}\big|\psi_{\bar{a}\bm{\lambda}}}}{\braket{\pi_{\omega}\big|L_{j}\big|\psi_{\bar{a}\bm{\lambda}}}}\nonumber \\
 & =\frac{\braket{\psi_{\bar{a}\bm{\lambda}}\big|L_{i}\big|\pi_{\omega}}}{\braket{\psi_{\bar{a}\bm{\lambda}}\big|L_{\bar{i}}\big|\pi_{\omega}}}\frac{\braket{\psi_{\bar{a}\bm{\lambda}}\big|L_{\bar{i}}\big|\pi_{\omega}}}{\braket{\psi_{\bar{a}\bm{\lambda}}\big|L_{j}\big|\pi_{\omega}}}=\frac{\braket{\pi_{\omega}\big|L_{i}\big|\psi_{c\lambda}}}{\braket{\pi_{\omega}\big|L_{\bar{i}}\big|\psi_{c\bm{\lambda}}}}\frac{\braket{\pi_{\omega}\big|L_{\bar{i}}\big|\psi_{c\bm{\lambda}}}}{\braket{\pi_{\omega}\big|L_{j}\big|\psi_{c\bm{\lambda}}}}=\frac{\braket{\pi_{\omega}\big|L_{i}\big|\psi_{c\lambda}}}{\braket{\pi_{\omega}\big|L_{j}\big|\psi_{c\bm{\lambda}}}}.
\end{align}
In othe words, 
\begin{equation}
\eta_{\omega,\,ij}^{(b)}=\eta_{\omega,\,i\bar{i}}^{(b)}\eta_{\omega,\,\bar{i}j}^{(b)}=\eta_{\omega,\,i\bar{i}}^{(\bar{a})*}\eta_{\omega,\,\bar{i}j}^{*(\bar{a})}=\eta_{\omega,\,i\bar{i}}^{(\bar{a})}\eta_{\omega,\,\bar{i}j}^{(\bar{a})}=\eta_{\omega,\,i\bar{i}}^{(c)*}\eta_{\omega,\,\bar{i}j}^{(c)*}=\eta_{\omega,\,i\bar{i}}^{(c)*}\eta_{\omega,\,\bar{i}j}^{(c)*}=\eta_{\omega,\,ij}^{(c)*},
\end{equation}
which is equivalent to Eq.~\eqref{eq:W-opt}. Similar with the previous case, if $\eta_{\omega,\,i\bar{i}}^{(b)}$ and $\eta_{\omega,\,\bar{i}j}^{(b)}$ are of the type $0/0$, it will invalidate the above identities. Nevertheless, 
Observation~\ref{obs:non-singular} excludes these singular cases.

\section{Results on the linearly (in)dependence of $\{W_{ij,\,ab}^{(\alpha)},\,M_{i,ab}^{(\alpha)}\}$}

\label{sec:lin-indept}

In this section, unless otherwise stated, we always focus on real
vector space and therefore linear independence or dependence is always
defined over the real field. 
\begin{lem}
\textup{The QFIM is non-singular, then the set of $\{L_{i}\}_{i=1}^{s}$
are linearly independent in $\mathrm{Herm}(d)$.}
\end{lem}
\begin{proof}
We prove by contradiction. Consider the matrix element of the QFIM
$[F^{Q}]_{jk}=\mathrm{ReTr}(\rho_{\bm{\lambda}}L_{j}L_{k})$. If $\{L_{i}\}_{i=1}^{s}$
are not linearly dependent, then $\exists k$ such that $L_{k}=\sum_{l\neq k}c_l L_{l}$ with $c_l\in \mathbb{R}$
and therefore $[F^{Q}]_{jk}=\sum_{l\neq k} c_l [F^{Q}]_{jl}$. That is, the
$k$-column of the QFIM is linear combination of the other column.
Therefore the QFIM is rank-deficient and singular, which completes
the proof. 
\end{proof}
\begin{lem}
\textup{If the QFIM is non-singular, then $[F^{Q}]_{jj}=\mathrm{Tr}(\rho_{\bm{\lambda}}L_{j}^{2})$
must be strictly positive.}
\end{lem}
\begin{proof}
We prove by contradiction. If $\exists j$ such that $[F^{Q}]_{jj}=\mathrm{Tr}(\rho_{\bm{\lambda}}L_{j}^{2})=0$,
the using Cauch-Schwarz inequality
\begin{equation}
|\mathrm{Tr}(\rho_{\bm{\lambda}}L_{j}L_{k})|^{2}=|\langle L_{j}\sqrt{\rho_{\bm{\lambda}}},\,L_{k}\sqrt{\rho_{\bm{\lambda}}}\rangle|^{2}\leq F_{jj}^{Q}F_{kk}^{Q}=0,
\end{equation}
which forces $\mathrm{Tr}(\rho_{\bm{\lambda}}L_{j}L_{k})=0$ for all
$k$. This means that $[F^{Q}]_{jk}=0$ for all $k$, which indicates
that the $j$-row and the $j$-column of the QFIM vanishes. Therefore
the QFIM is rank-deficient and singular, which completes the proof.
\end{proof}
An immediate consequence of above lemma is the following:
\begin{lem}
\textup{If the QFIM is non-singular, then for every parameter $\lambda_{j}$,
there must exists an eigenstate of $\rho_{\bm{\lambda}}$ corresponding
to a positive eigenvalue such that $L_{j}\ket{\psi_{a\bm{\lambda}}}\neq0$.}
\end{lem}
\begin{proof}
It can be readily proved by contraction. We note that $[F^{Q}]_{jj}=\mathrm{Tr}(\rho_{\bm{\lambda}}L_{j}^{2})=\sum_{a}p_{a\bm{\lambda}}\|L_{j}\ket{\psi_{a\bm{\lambda}}}\|^{2}$.
Thus for given $j$, if $L_{j}\ket{\psi_{a\bm{\lambda}}}=0$ for all
$a$, then $[F^{Q}]_{jj}=0$, which leads to singular QFIM. 
\end{proof}

Before proceeding, let us  note that
the SLD operator can be rewritten as~$L_{i}=L_{i}^{\mathrm{ss}}+(L_{i}^{\mathrm{sk}}+\mathrm{h.c.})+L_{i}^{\mathrm{kk}}$,
where $L_{i}^{\alpha\beta}\equiv\Pi_{\mathrm{\alpha}}L_{i}\Pi_{\mathrm{\beta}}$,
$\Pi_{\mathrm{s}}$ and $\Pi_{\mathrm{k}}$ are projectors to the
support and kernel of $\rho_{\bm{\lambda}}$, respectively. In particular,
it is shown in \citep{yang2019optimal,yang2024purestate} that 
\begin{equation}
L_{i}^{\mathrm{sk}}=2\sum_{a}\ket{\psi_{a\bm{\lambda}}}\bra{\partial_{i}\psi_{a\bm{\lambda}}}\Pi_{k},\quad L_{i}^{\mathrm{kk}}=0.\label{eq:Lsk-expression}
\end{equation}

\begin{thm} \label{thm:vanishing-M-condition}
\textup{(Vanishing $M$ condition)
For a given estimation parameter $\lambda_{j}$, if $M_{j,\,ab}^{(\alpha)}=0$
for all $a$, $b$, $\alpha$. Then $L_{j}$ bears the structure
\begin{equation}
L_{j}=\alpha_{j}\Pi_{\mathrm{s}},
\end{equation}
where $\alpha_{j}\in\mathbb{R}$ and the density matrix satisfies}
\begin{equation}
\partial_{j}\ln p_{a\bm{\lambda}}=\alpha_{j},\,\forall a\in[1,r],
\end{equation}
\begin{equation}
p_{a\bm{\lambda}}=p_{b\bm{\lambda}}\,\mathrm{or}\,\braket{\partial_{j}\psi_{a\bm{\lambda}}|\psi_{b\bm{\lambda}}}=0,\,\forall a,\,b\in[1,r],a\neq b,
\end{equation}
\textup{and $\ket{\partial_{j}\psi_{a\bm{\lambda}}}$ fully lies the
support of $\rho_{\bm{\lambda}}$.}
\end{thm}
\begin{proof}
We note that the set of matrices $\{\sigma_{ab}^{(\alpha)}\}_{}$
containing $r^{2}$ elements forms a basis of the subspace of a Hermitian
operators that fully lies on the support of $\rho_{\lambda}$, which
is isomorphic to $\mathrm{Herm}(r)$ or $\mathfrak{u}(r)$. Note that $\sigma_{ab}^{(\alpha)}$ is an operator acting on the support of $\rho_{\bm{\lambda}}$, we know $L_{i}^{\mathrm{sk}}\sigma_{ab}^{(\alpha)}=\sigma_{ab}^{(\alpha)}L_{i}^{\mathrm{ks}}=0$. Therefore, we obtain
\begin{equation}
[\sigma_{ab}^{(\alpha)},\,L_{i}]=[\sigma_{ab}^{(\alpha)},\,L_{i}^{\mathrm{ss}}]+(\sigma_{ab}^{(\alpha)}L_{i}^{\mathrm{sk}}-\mathrm{h.c.}).\label{eq:commutator}
\end{equation}
On the other hand, we know $[\sigma_{ab}^{(\alpha)},\,L_{i}^{\mathrm{ss}}]$is an operator acting on the support of $\rho_{\bm{\lambda}}$, $\sigma_{ab}^{(\alpha)}L_{i}^{\mathrm{sk}}$ maps the kernel space to the support space of $\rho_{\bm{\lambda}}$ while its Hermitian conjugate maps support space to the kernel space of $\rho_{\bm{\lambda}}$, these three terms must be linearly independently. Using Eq.~\eqref{eq:Lsk-expression}, we obtain
\begin{equation}
L_{i}^{\mathrm{sk}}=0, \,[\sigma_{ab}^{(\alpha)},\,L_{i}^{\mathrm{ss}}]=0
\end{equation}
On the other hand, $\{\mathrm{i}\sigma_{ab}^{(\alpha)}\}_{ab,\,{\alpha}}$ is basis of the fundamental representation of  the $\mathfrak{u}(r)$ algebra, which is irreducible.
According to Schur's lemma, we know 
\begin{equation}
L_{i}^{\mathrm{ss}}\propto\Pi_{\mathrm{s}}.
\end{equation}
The Hermiticity requirement imposes the proportionality constant must
be real. This concludes the proof.
\end{proof}

The following theorem states that for a given $i$, when the set of vectors $\{\Pi_{\mathrm{k}}\ket{\partial_{i}\psi_{a\bm{\lambda}}}\}_{a\in[1,\,r]}$
are linearly independent, the set of matrices $\{M_{i,\,ab}^{(\alpha)}\}_{a,\,b\in[1,\,r],\,\alpha=x,y,z}$
are also linearly independent.
\begin{thm}
\textup{(Linear independence among $M$'s) For a given $i$, if the
set of vectors$\{\Pi_{k}\ket{\partial_{i}\psi_{a\bm{\lambda}}}\}_{a\in[1,r]}$
is linearly independent then the set of matrices $\{M_{i,\,ab}^{(\alpha)}\}$
must be linear independent.}
\end{thm}
\begin{proof}
We consider the linear combination of the $M$ matrices as follows:
\begin{align}
\sum_{ab,\,\alpha}c_{i,\,ab}^{(\alpha)}M_{i,\,ab}^{(\alpha)} & =[\sum_{ab,\,\alpha}c_{i,\,ab}^{(\alpha)}\sigma_{ab}^{(\alpha)},\,L_{i}]\nonumber \\
 & =[\sum_{ab,\,\alpha}c_{i,\,ab}^{(\alpha)}\sigma_{ab}^{(\alpha)},\,L_{i}^{\mathrm{ss}}]+\left(\sum_{ab,\,\alpha}c_{i,\,ab}^{(\alpha)}\sigma_{ab}^{(\alpha)}L_{i}^{\mathrm{sk}}-\mathrm{h.c.}\right),
\end{align}
where the coefficient $c_{i,\,ab}^{(\alpha)}\in \mathbb{R}$. 
More explicitly, according to Eq.~(\ref{eq:Lsk-expression}), we find 
\begin{equation}
\sum_{ab,\,\alpha}c_{i,\,ab}^{(\alpha)}\sigma_{ab}^{(\alpha)}L_{i}^{\mathrm{sk}}=\left(\sum_{ab,\,\alpha}c_{i,\,ab}^{(\alpha)}\sigma_{ab}^{(\alpha)}\right)\sum_{c}\ket{\psi_{c\bm{\lambda}}}\bra{\partial_{i}\psi_{c\bm{\lambda}}}\Pi_{k}.
\end{equation}
Since the set of vectors$\{\Pi_{k}\ket{\partial_{i}\psi_{a\bm{\lambda}}}\}_{a\in[1,\,r]}$
lie in the kernel of $\rho_{\bm{\lambda}}$ and are linearly independent and $\ket{\psi_{c\bm{\lambda}}}$ lies on the support of $\rho_{\bm{\lambda}}$,
in order for the commutator $[\sum_{ab,\,\alpha}c_{i,\,ab}^{(\alpha)}\sigma_{ab}^{(\alpha)},\,L_{i}]$
to vanish, one must have
\begin{equation}
\left(\sum_{ab,\,\alpha}c_{i,\,ab}^{(\alpha)}\sigma_{ab}^{(\alpha)}\right)\ket{\psi_{c\bm{\lambda}}}=0,\,\forall c.
\end{equation}
This further leads to
\begin{equation}
\sum_{ab,\,\alpha}c_{i,\,ab}^{(\alpha)}\sigma_{ab}^{(\alpha)}=0.
\end{equation}
Since the set of operators $\{\sigma_{ab}^{(\alpha)}\}$ are linearly
independent, this means that $c_{i,\,ab}^{(\alpha)}=0$, which concludes
the proof.
\end{proof}

On the other hand, if $L_{i}^{\mathrm{sk}}=0$,
i.e., the SLD $L_{i}$ fully lies on the support of $\rho_{\bm{\lambda}}$,
the set of matrices $\{M_{i,\,ab}^{(\alpha)}\}_{a,\,b\in[1,\,r],\,\alpha=x,y,z}$
must be linearly dependent, as stated in the following theorem:
\begin{thm}
\textup{(Linear dependence among $M$'s) If $L_{i}^{\mathrm{sk}}=0$, i.e., $L_{i}$ only lies on the
support of $\rho_{\bm{\lambda}}$ (including the case of full rank
states), $\{M_{i,\,ab}^{(\alpha)}\}$ is linearly dependent.}
\end{thm}
\begin{proof}
We note that the set of basis $\{\sigma_{ab}^{(\alpha)}\}_{}$ containing
$r^{2}$ elements nows forms a basis of a Hermitian operators that
fully lies on the support of $\rho_{\lambda}$. Therefore, there exists
a set of real numbers $\{c_{i,\,ab}^{(\alpha)}\}$, with at least
one of them are non-zero, such that 
\begin{equation}
L_{i}=\sum_{ab,\,\alpha}c_{i,\,ab}^{(\alpha)}\sigma_{ab}^{(\alpha)},
\end{equation}
which leads to
\begin{equation}
\sum_{ab,\,\alpha}c_{i,\,ab}^{(\alpha)}M_{i,\,ab}^{(\alpha)}=0.
\end{equation}
Therefore in this case $\{M_{i,\,ab}^{(\alpha)}\}$ are linearly dependent.
\end{proof}

\section{Details about the efficient numerical search algorithms for global
saturation}

\label{sec:justification-efficient-search-algorithm}

The numerical search algorithm presented in the main text is justified
by the following theorem:
\begin{thm}
\textup{(Global saturability criterion) \label{thm:global-sat}When
PCC is satisfied, the QCRB is saturable }if and only if\textup{
there exists a tuplle $\Lambda\equiv\{\bm{v}^{(\omega)},\,\alpha^{(\omega)}\in[0,1]\}_{\omega\in\Omega}$,
where $\bm{v}^{(\omega)}$ is defined as in Theorem~\ref{thm:outcome-wise=000020sat},
such that 
\begin{equation}
\sum_{\omega}\alpha^{(\omega)}=d,\,\sum_{\omega}\alpha^{(\omega)}\bm{v}^{(\omega)}=0.\label{eq:global-sat}
\end{equation}
}
\end{thm}
\begin{proof}
Consider the parameterization of the optimal rank-one POVM given by
\[
E_{\omega}=\alpha^{(\omega)}\Pi_{\omega}=\frac{\alpha^{(\omega)}}{d}(\mathbb{I}-\bm{v}^{(\omega)}\cdot\bm{T}).
\]
The completeness relation $\sum_{\omega}E_{\omega}=\mathbb{I}$ implies
that 
\begin{equation}
\sum_{\omega}E_{\omega}=\frac{1}{d}\left(\sum_{\omega}\alpha^{(\omega)}\right)\mathbb{I}+\frac{1}{d}\sum_{k=1}^{n}\left(\sum_{\omega}\alpha^{(\omega)}v_{k}^{(\omega)}\right)T_{k}=\mathbb{I}.\label{eq:E-sum}
\end{equation}
Taking trace on both sides leads to 
\begin{equation}
\sum_{\omega}\alpha^{(\omega)}=d,\label{eq:cond1}
\end{equation}
which leads to 
\begin{equation}
\sum_{k=1}^{n}\left(\sum_{\omega}\alpha^{(\omega)}v_{k}^{(\omega)}\right)T_{k}=0.
\end{equation}
Multiplying above equation from the left by $T_{l}$ and then taking
the trace yields
\begin{equation}
\sum_{\omega}\alpha^{(\omega)}v_{k}^{(\omega)}=0. \label{eq:cond2}
\end{equation}
On the other hand, it can be readily checked that with Eqs.~(\ref{eq:cond1},~\ref{eq:cond2}),
$\sum_{\omega}E_{\omega}=\mathbb{I}$ is satisfied according to Eq.~(\ref{eq:E-sum}).
Furthermore, for projective measurements, $|\Omega|=d$ and
$\alpha^{(\omega)}=1$, we only require 
\begin{equation}
\sum_{\omega=1}^{d}v_{k}^{(\omega)}=0.
\end{equation}
\end{proof}

\section{Detailed proof of Theorem~\ref{thm:sufficiency}}

\label{sec:Detailed-proof-sufficiency}
\begin{lem}
\textup{\label{lem:rank-one-generalied-structure}Given a projector
$\mathbb{I}^{(\mu+1)}\equiv\mathbb{I}-\sum_{q=1}^{\mu}\Pi_{q}$, where
$\Pi_{q}\equiv\ket{\pi_{q}}\bra{\pi_{q}}$ and $\{\ket{\pi_{q}}\}_{q=1}^{\mu}$
is an orthonormal basis where $\mu\le d-2$, and a set of traceless
Hermitian matrices $\{T_{k}^{(\mu+1)}\}_{k\ge1}$, then the }real\textup{
linear combination between $\mathbb{I}^{(\mu+1)}$ and $\{T_{k}^{(\mu+1)}\}_{k\ge1}$
forms a rank-one projector $\Pi_{\mu+1}=\ket{\pi_{\mu+1}}\bra{\pi_{\mu+1}}$
that is orthogonal to $\{\Pi_{q}\}_{q=1}^{\mu}$ }if and only if\textup{
there exists a real vector $\bm{v}$ such that (i) The kernel space
of $\bm{v}^{(\mu+1)}\cdot\bm{T}^{(\mu+1)}\equiv\sum_{k}v_{k}^{(\mu+1)}T_{k}^{(\mu+1)}$
is $\mathrm{span}_{\mathbb{C}}\{\ket{\pi_{q}}\}_{q=1}^{\mu}$. (ii)
$1$ is the eigenvalue of $\bm{v}^{(\mu+1)}\cdot\bm{T}^{(\mu+1)}$
with degeneracy $d-\mu-1$. }
\end{lem}
\begin{proof}
The proof is similar to Theorem~\ref{thm:outcome-wise=000020sat}
in the main text. The backward direction is straightforward to prove.
We prove the forward direction. Without loss of generality, we take
\begin{equation}
\Pi_{\mu+1}=\frac{1}{d-\mu}\left(\alpha\mathbb{I}^{(\mu+1)}-\bm{v}^{(\mu+1)}\cdot\bm{T}^{(\mu+1)}\right)
\end{equation}
with $\alpha\in\mathbb{R}$, from which we conclude 
\begin{equation}
\bm{v}^{(\mu+1)}\cdot\bm{T}^{(\mu+1)}=\alpha\mathbb{I}^{(\mu+1)}-(d-\mu)\ket{\pi_{\mu+1}}\bra{\pi_{\mu+1}}=\alpha\mathbb{I}-\alpha\sum_{q=1}^{\mu}\Pi_{q}-(d-\mu)\ket{\pi_{\mu+1}}\bra{\pi_{\mu+1}}.
\end{equation}
The traceless property of $T_{k}^{(\mu+1)}$ implies that $\alpha=1$,
which is then implies 
\begin{equation}
\bm{v}^{(\mu+1)}\cdot\bm{T}^{(\mu+1)}=\mathbb{I}-\sum_{q=1}^{\mu}\Pi_{q}-(d-\mu)\ket{\pi_{\mu+1}}\bra{\pi_{\mu+1}}.
\end{equation}
The r.h.s gives the spectral properties of $\bm{v}^{(\mu+1)}\cdot\bm{T}^{(\mu+1)}:$
The kernel space is spanned by $\{\ket{\pi_{q}}\}_{q=1}^{\mu}$ while
$(d-\mu-1)$ non-zero eigenvalues are equal to $1$.
\end{proof}
\begin{obs} \label{obs:constraint-counting1}\textup{Given a set of orthonormal
$d$-dimension complex bases$\{\ket{\pi_{q}}\}_{q=1}^{\mu}$ and a
linear Hermitian operator $T$ satisfying $\braket{\pi_{q}|T|\pi_{q}}=0$
for $q=1,2,\,\cdots\mu$. The linear subspace $\mathrm{span}_{\mathbb{C}}\{\ket{\pi_{q}}\}_{q=1}^{\mu}$
becomes the kernel of $T$ \textit{if and only if} $T$ are $\mu$-fold
degenerate and the corresponding eigenvectores lies in $\mathrm{span}_{\mathbb{C}}\{\ket{\pi_{q}}\}_{q=1}^{\mu}$
.}

\end{obs}
\begin{proof}
The forward direction is straightforward. To prove the backward direction,
we denote the degenerate eigenvalue as $\lambda$ and a basis in the
degenerate subspace as $\ket{\lambda_{q}}_{q=1}^{\mu}$ . It can be
readily observed that 
\begin{equation}
T\sum_{q=1}^{\mu}\ket{\lambda_{q}}\bra{\lambda_{q}}=T\sum_{q=1}^{\mu}\ket{\pi_{q}}\bra{\pi_{q}}=\lambda\sum_{q=1}^{\mu}\ket{\lambda_{q}}\bra{\lambda_{q}}.
\end{equation}
Then taking trace on both sides, one conclude $\lambda=0$, i.e.,
$\mathrm{span}_{\mathbb{C}}\{\ket{\pi_{q}}\}_{q=1}^{\mu}$ is the
kernel of $T$.
\end{proof}
\begin{obs} \label{obs:constraint-counting2}\textup{Given a set of $d$-dimensional
complex orthonormal basis $\{\ket{\pi_{1}},\,\ket{\pi_{2}}\cdots\ket{\pi_{\mu}}\}$,
the effective \textit{real} constraints of requiring a normalized
vector $\ket{\lambda}\in\mathrm{span}_{\mathbb{C}}\{\ket{\pi_{1}},\,\ket{\pi_{2}}\cdots\ket{\pi_{\mu}}\}$
at most $2(d-\mu)$.}
\end{obs} 

\begin{proof}
It can be readily observed that $\ket{\lambda}\in\mathrm{span}_{\mathbb{C}}\{\ket{\pi_{1}},\,\ket{\pi_{2}}\cdots\ket{\pi_{\mu}}\}$
is equivalent to $\ket{\lambda}$ is orthogonal to the remaining $d-\mu$
orthonormal vectors $\{\ket{\pi_{\mu+1}},\,\ket{\pi_{\mu+2}},\,\cdots\ket{\pi_{d}}\}$,
which imposes at most $d-\mu$ complex constraints. 
\end{proof}
Now we are in a position to give a comprehensive proof of Theorem~\ref{thm:sufficiency}
in the main text.
\begin{proof}
The necessity of PCC has shown in Ref.~\citep{yang2019optimal}.
Here we demonstrate its sufficiency. 
\begin{enumerate}
\item The $1^{\mathrm{st}}$ projective rank-one POVM can be constructed
according to Theorem~\ref{thm:outcome-wise=000020sat} in the main
text. More specifically,
\begin{enumerate}
\item Define $\mathcal{V}_{\perp}^{(1)}=\mathcal{V}_{\perp}$ and $T_{k}^{(1)}=T_{k},\,k\ge1$
and satisfy the normalization condition~(\ref{eq:su-alg-norm}).
Denote the eigenvalues and eigenvectors of $\bm{v}^{(1)}\cdot\bm{T}^{(1)}=\sum_{k=1}^{n}v_{k}^{(1)}T_{k}^{(1)}$
as $\{\pi_{q}(\bm{v}^{(1)})\}_{q=1}^{d}$ and $\{\ket{\pi_{q}(\bm{v}^{(1)})}\}_{q=1}^{d}$
respectively.
\item In order for the linear combination between $\mathbb{I}$ and $T_{k}^{(1)}$
to form a rank-one projector, according to Theorem~\ref{thm:outcome-wise=000020sat}
in the main text, we must impose $(d-1)$ eigenvalues of $\bm{v}^{(1)}\cdot\bm{T}^{(1)}$
be equal to $1$. With out loss of generality can be taken to be $\pi_{q}(\bm{v}^{(1)})=1$
for $2\le q\leq d$. Then the first rank-one optimal projector becomes
\begin{equation}
\Pi_{1}=\frac{1}{d}(\mathbb{I}-\bm{v}^{(1)}\cdot\bm{T}^{(1)})=\ket{\pi_{1}(\bm{v}^{(1)})}\bra{\pi_{1}(\bm{v}^{(1)})},
\end{equation}
with $v_{k}^{(1)}=-\braket{\pi_{1}(\bm{v}^{(1)})|T_{k}^{(1)}|\pi_{1}(\bm{v}^{(1)})}$
and $\sum_{k=1}^{n}[v_{k}^{(1)}]^{2}=d-1$.
\end{enumerate}
\item The remaining optimal rank-one projectors can be constructed iteratively.
To construct the $(\mu+1)^{\mathrm{th}}$ projective rank-one POVM:
\begin{enumerate}
\item $\Pi_{\mu+1}$ must lies in the orthogonal complement of the linear
subspace $\mathrm{span}_{\mathbb{R}}\{\Pi_{q}\}_{q=1}^{\mu}$, which
we denoted as $\mathcal{V}_{\perp}^{(\mu+1)}$. We define $\mathbb{I}^{(\mu+1)}\equiv\mathbb{I}^{(\mu)}-\Pi_{\mu}=\mathbb{I}-\sum_{q=1}^{\mu}\Pi_{q}\in\mathcal{V}_{\perp}^{(\mu+1)}$.
The other $n-\mu$ mutually orthogonal bases on $\mathcal{V}_{\perp}^{(\mu+1)}$
are denoted as $T_{k}^{(\mu+1)}$, where $k=1,\,2,\cdots n-\mu$.
It can be readily observed that 
\begin{equation}
\langle T_{k}^{(\mu+1)},\Pi_{q}\rangle=\braket{\pi_{q}(\bm{v}^{(q)})|T_{k}^{(\mu+1)}|\pi_{q}(\bm{v}^{(q)})}=0,\,q\in[1,\,\mu].
\end{equation}
Furthermore, they must be traceless since
\begin{equation}
\langle T_{k}^{(\mu+1)},\,\mathbb{I}^{(\mu+1)}\rangle=\langle T_{k}^{(\mu+1)},\,\mathbb{I}\rangle=0.
\end{equation}
We choose the normalization condition 
\begin{equation}
\langle T_{k}^{(\mu+1)},\,T_{l}^{(\mu+1)}\rangle=(d-\mu)\delta_{kl}.
\end{equation}
We denote the eigenvalues and eigenvectors of $\bm{v}^{(\mu+1)}\cdot\bm{T}^{(\mu+1)}=\sum_{k=1}^{n-\mu}v_{k}^{(\mu+1)}T_{k}^{(\mu+1)}$
as $\{\pi_{q}(\bm{v}^{(\mu+1)})\}_{q=1}^{d}$ and $\{\ket{\pi_{q}(\bm{v}^{(\mu+1)})}\}_{q=1}^{d}$
respectively.
\item In order to form a rank-one projector from the linear combinations
between $\mathbb{I}^{(\mu+1)}$ and $T_{k}^{(\mu+1)}$, according
to Lemma~\ref{lem:rank-one-generalied-structure}, without loss of
generality, we require (i) $\{\ket{\pi_{q}(\bm{v}^{(q)})}\}_{q=1}^{\mu}$
lies in the kernel of $\bm{v}^{(\mu+1)}\cdot\bm{T}^{(\mu+1)}$ (ii)
$(d-\mu-1)$ eigenvalues equal $1$. According to Observation~(\ref{obs:constraint-counting1}),
condition (i) can be guaranteed by imposing the $\mu$ eigenvalues
are equal, i.e.,
\begin{equation}
\pi_{1}(\bm{v}^{(\mu+1)})=\pi_{2}(\bm{v}^{(\mu+1)})=\cdots\pi_{\mu}(\bm{v}^{(\mu+1)}),
\end{equation}
and $\ket{\pi_{q}(\bm{v}^{(\mu+1)})}\in\mathrm{span}\{\ket{\pi_{q}(\bm{v}^{(q)})}\}_{q=1}^{\mu}$
for $1\le q\leq\mu$. Condition (ii) can be guaranteed by imposing
$\pi_{q}(\bm{v}^{(\mu+1)})=1$ for $\mu+2\le q\leq d$. Then 
\begin{equation}
\Pi_{\mu+1}=\frac{1}{d-\mu}(\mathbb{I}^{(\mu+1)}-\bm{v}^{(\mu+1)}\cdot\bm{T}^{(\mu+1)})=\ket{\pi_{\mu+1}(\bm{v}^{(\mu+1)})}\bra{\pi_{\mu+1}(\bm{v}^{(\mu+1)})},
\end{equation}
with $v_{k}^{(\mu+1)}=-\braket{\pi_{\mu+1}(\bm{v}^{(\mu+1)})|T_{k}^{(\mu+1)}|\pi_{\mu+1}(\bm{v}^{(\mu+1)})}$
and $\sum_{k=1}^{n-\mu}[v_{k}^{(\mu+1)}]^{2}=d-\mu-1$.
\end{enumerate}
\end{enumerate}
Now let us counter the constraints and free-parameters in the construction
of $(\mu+1)$-th rank-one projectors: According to Observations~\ref{obs:constraint-counting1}
and \ref{obs:constraint-counting2}, it is clear that condition (i)
imposes $\mu(d-\mu)+\mu-1$ real constraints. Condition (ii) imposes
$d-\mu-1$ real constraints. On the other hand, the number of real
free parameters is $n-\mu$. In order for real free parameters outnumbers
the real constraints for the construction of optimal rank-one projectors,
up to the $(d-1)$-th one, the inequality (\ref{eq:suff-ineq}) in
the main text must hold. 
\end{proof}

\end{document}